\documentclass[a4paper,aps,pre,superscriptaddress,floatfix,nofootinbib,longbibliography,notitlepage,twocolumn]{revtex4-1}

\usepackage[utf8]{inputenc}
\usepackage[english]{babel}

\usepackage{amsmath,amssymb,amsfonts,amsthm,mathrsfs,amsopn}
\usepackage{mathtools}
\usepackage[percent]{overpic}
\usepackage{graphicx}
\usepackage{xcolor}
\usepackage{bm}
\usepackage{braket}
\usepackage{pifont}
\usepackage{blkarray,bigstrut}
\usepackage{multirow}
\usepackage{array}
\usepackage{booktabs}
\usepackage{dcolumn}
\usepackage{float}
\usepackage{svg}
\usepackage{comment}
\usepackage{verbatim}

\usepackage{hyperref}

\DeclareMathOperator{\Arg}{Arg}

\def\be{\begin{equation}}
\def\ee{\end{equation}}
\def\bc{\begin{center}}
\def\ec{\end{center}}
\def\bea{\begin{eqnarray}}
\def\eea{\end{eqnarray}}

\theoremstyle{plain}
\newtheorem{theorem}{Theorem}
\newtheorem{proposition}[theorem]{Proposition}

\newtheorem{remark}[theorem]{Remark}

\definecolor{brickred}{rgb}{0.7, 0.25, 0.33}
\definecolor{applegreen}{rgb}{0.55, 0.71, 0.0}

\begin{document}
\title{The role of asymmetric time delay and its structure in 1D swarmalators}

\author{Rommel Tchinda Djeudjo}
\email{rommel.tchindadjeudjo@unamur.be}
\affiliation{Department of Mathematics \& naXys, Namur Institute for Complex Systems, University of Namur, Rue Grafé 2, B-5000 Namur, Belgium}

\author{Gourab Kumar Sar}
\affiliation{Complexity Science Group, Department of Physics and Astronomy, University of Calgary, Calgary, AB, Canada}

\author{Timoteo Carletti}
\email{timoteo.carletti@unamur.be}
\affiliation{Department of Mathematics \& naXys, Namur Institute for Complex Systems, University of Namur, Rue Grafé 2, B-5000 Namur, Belgium}

\begin{abstract}
Swarmalators are a class of coupled oscillators that simultaneously synchronize in both space and phase, providing a minimal model for systems ranging from biological microswimmers to robotic swarms. Time delay is ubiquitous in such systems, arising from finite signal propagation speeds and sensory processing lags, yet its structural form, whether symmetric or asymmetric, has received little attention. Here, we study a one-dimensional swarmalator model with asymmetric time delay, in which the delay enters only the self-interaction terms of the spatial and phase dynamics, breaking the symmetry assumed in prior work. We identify various collective states such as async, static phase wave, static $\pi$, and active $\pi$, and derive analytical stability boundaries for each as a function of the coupling parameters and delay. Our analysis reveals that the asymmetric delay structure fundamentally reshapes the collective phase diagram: in particular, for the asymmetric delay models, increasing the delay systematically expands the active $\pi$ state at the expense of other ordered states, in contrast to the symmetric delay model, which more strongly promotes the presence of unsteady states that are generally not well ordered. By providing closed-form stability conditions validated against numerical simulations, our work establishes that the internal structure of the delay, not merely its magnitude, is a decisive factor in determining the emergent collective behavior of swarmalator populations. 
\end{abstract}
\maketitle

\section{Introduction}
Many natural and engineered systems exhibit spatial self-organization, moreover systems possessing some internal degree of freedom, can also synchronize the latter variables with the spatial ones. Japanese tree frogs reposition themselves around a pond while entraining their calling rhythms \cite{aihara2014spatio}. Starfish embryos assemble into spinning \textit{living crystals} by coupling their swimming trajectories to their internal beating cycles \cite{tan2022odd}. Magnetic domain walls, active colloids, and sperm cells align both their direction of motion and their phase of oscillation \cite{hrabec2018velocity, yan2016reconfiguring, riedel2005self}, and modern robotic swarms are increasingly designed to coordinate location and internal state in parallel \cite{barcis2020sandsbots}. The minimal framework that captures such joint space-phase organization is the \textit{swarmalator} model introduced by O'Keeffe et al. \cite{o2017oscillators}, in which Kuramoto-type phase oscillators are bidirectionally coupled to their own spatial dynamics so that synchronization and aggregation feed back on one another.

Among the various formulations~\cite{sar2026interplay}, the one-dimensional (1D) variant~\cite{o2022collective} has emerged as a particularly useful benchmark. The 1D model of swarmalators is governed by the pair of equations
\begin{align}
    \dot{x}_i &= \frac{J}{N} \sum_{j = 1}^{N}
        \sin \bigl(x_j - x_i\bigr)\,
        \cos \bigl(\theta_j - \theta_i\bigr), 
    \\
    \dot{\theta}_i &= \frac{K}{N} \sum_{j = 1}^{N}
        \sin \bigl(\theta_j - \theta_i\bigr)\,
        \cos \bigl(x_j - x_i\bigr)\, .
\end{align}
By placing both the position $x_i$ and the phase $\theta_i$ on the circle $\mathbb{S}^1$, the model retains a rich phenomenology of the two-dimensional model (2D)~\cite{o2017oscillators}, while being amenable to analytical treatment. This balance between tractability and richness has made the 1D swarmalator model a natural starting point for exploring how additional physical ingredients reshape collective behavior. Building on this framework, several extensions have been investigated, by including swarmalators with nonidentical velocities and frequencies~\cite{yoon2022sync}, disordered coupling \cite{o2022swarmalators,hao2023attractive}, random pinning \cite{sar2023pinning,sar2024solvable,sar2023swarmalators}, external forcing \cite{anwar2024forced,anwar2025forced}, phase frustration \cite{lizarraga2023synchronization}, and other related effects and analytical techniques~\cite{sar2025effects,hong2023swarmalators,lr2r-ynzs,anwar2024collective,senthamizhan2026swarmalators, sar2025strategy,o2025global,o2025stability}.

Another ingredient, ubiquitous in both biological and engineered systems, is time delay. Finite signal propagation speeds, sensory integration times, and communication or actuation latencies all prevent agents from responding instantaneously to one another. The consequences of delay are well documented in the pure phase-oscillator setting: in the Kuramoto model, even modest delays destabilize the incoherent state in nontrivial ways, generate multistability between synchronized branches, and can induce oscillation death \cite{yeung1999time,schuster1989mutual,earl2003synchronization}. For swarmalators, by contrast, the effect of delay has only recently begun to be explored. Initial studies have focused on the 2D model, where numerical investigations have revealed that delay can qualitatively alter the phase diagram and give rise to new collective states~\cite{blum2024swarmalators,kumpeerakij2025aging}; however, the full 2D setting is analytically demanding, as the coupled delay-differential equations for position and phase do not generally admit closed-form stability conditions, motivating thus the study of lower-dimensional variants in which the essential effects of delay can be isolated. A first step in this direction has been taken very recently where a symmetrically delayed version of the 1D model -- in which the same delay enters every coupling channel -- and showed that delay substantially reorganizes the phase diagram~\cite{o2026time}.

However, physical interactions are rarely uniformly delayed across all coupling channels. The alignment-generating terms and the modulating factors that gate them typically originate from distinct physical processes and are subject to different latencies: an agent may receive a delayed estimate of its neighbors' phases through finite-speed sensing while evaluating their spatial proximity with negligible lag, or, in a designed system, a latency may be deliberately introduced in one channel but not in another. Whether the structural form of the delay (which terms carry it, and which do not) matters for the emergent dynamics, or only its magnitude does, has not, to the best of our knowledge, been addressed for swarmalators.

In this paper, we take a first step in this direction by analyzing a 1D swarmalator model with asymmetric time delay, in which $\tau$ enters only the self-interaction (sine) terms of the spatial and phase equations, while the cross-modal (cosine) factors remain instantaneous. For completeness, we also examine the complementary asymmetric case, in which the delay enters only the cross-modal (cosine) factors and the self-interaction terms remain instantaneous (Appendix~\ref{appA}). Moreover, we compare both asymmetric variants with the symmetric delay model in which $\tau$ enters every coupling channel (Appendix~\ref{appB}). We identify the collective states supported by the model, derive closed-form stability boundaries for each as functions of the coupling parameters $(J,K)$ and $\tau$, and validate our analytical predictions against direct numerical simulations. A notable finding is that the delay can play relevant role of control parameter and, depending on where it is introduced, it can lead to dynamics that do not exist in the case without delay. For instance, when incorporated into the ``sine" terms, increasing the delay promotes the presence of the $\pi$ state, especially the active one. However, when incorporated into the ``cosine" terms, it promotes the emergence of unsteady or non-stationary states.

\section{Model }

We study the collective dynamics of a population of $N>>1$ interacting swarmalators, 
each characterized by a spatial position $x_i \in  \mathbb{S}^1$ and an internal 
phase $\theta_i \in \mathbb{S}^1$. Their coupled dynamics is governed by the 
following system of delay differential equations:
\begin{align}
    \dot{x}_i &= v_i + \frac{J}{N} \sum_{j = 1}^{N}
        \sin \bigl(x_j(t-\tau) - x_i(t)\bigr)
        \cos \bigl(\theta_j(t) - \theta_i(t)\bigr), 
    \label{eq:x_dot} \\
    \dot{\theta}_i &= \omega_i + \frac{K}{N} \sum_{j = 1}^{N}
        \sin \bigl(\theta_j(t-\tau) - \theta_i(t)\bigr)
        \cos \bigl(x_j(t) - x_i(t)\bigr).
    \label{eq:theta_dot}
\end{align}
Here, $(v_i,\omega_i)$ denote the natural velocity and frequency of the $i$-th swarmalator. 
In this study, however, we set them to zero for all swarmalators. The parameters 
$(J,K)$ represent the coupling constants. The spatial dynamics~\eqref{eq:x_dot} describe phase-dependent aggregation, whereas the phase 
dynamics~\eqref{eq:theta_dot} capture position-dependent synchronization. 
By anticipating on the following, we will show that this model can lead to a variety of collective behaviors.

To characterize the collective behavior of the system, we use a set of 
scalar observables and complex order parameters.

\paragraph*{Mean speed.}
To differentiate the stationary state from the non-stationary states, we measure the mean velocity $V$, defined as follows:
\begin{equation}
    V(t) = \frac{1}{N} \sum_{i=1}^{N} \|\mathbf{v}_i(t)\|, 
    \qquad 
    \|\mathbf{v}_i(t)\| = \sqrt{\dot{x}^2_i(t) + \dot{\theta}^2_i(t)}.
    \label{eq:avevel}
\end{equation}
A value $V \approx 0$ signals a \emph{static}  regime, whereas 
sustained positive values indicate persistent collective motion.

\paragraph*{Spatial--phase coherence.}
To quantify the correlation between spatial positions and internal phases, we 
use the two complex order parameters
\begin{equation}
    S_{\pm} e^{\iota \phi_{\pm}} = \frac{1}{N} \sum_{j=1}^{N} e^{\,\iota(x_j \pm \theta_j)}, \quad \iota =  \sqrt{-1}
    \label{eq:Spm}
\end{equation}
where $S_+$ and $S_-$ measure the degree of spatial-phase alignment. If the positions and the phases are perfectly correlated, i.e., $x_j = \pm \theta_j + c$ then $S_{\pm} = 1$.

\section{Results}

In this section, we present the main dynamical regimes observed in the swarmalator model with asymmetric delay. We first describe the different collective states obtained from numerical simulations and characterize them by using the metrics presented before. We then analyze the stability of the emerging states and derive analytical conditions for their existence and stability. We will eventually conclude by comparing the analytical findings with numerical simulations.

\begin{figure*}
    \centering
    \includegraphics[width=1.0\linewidth]{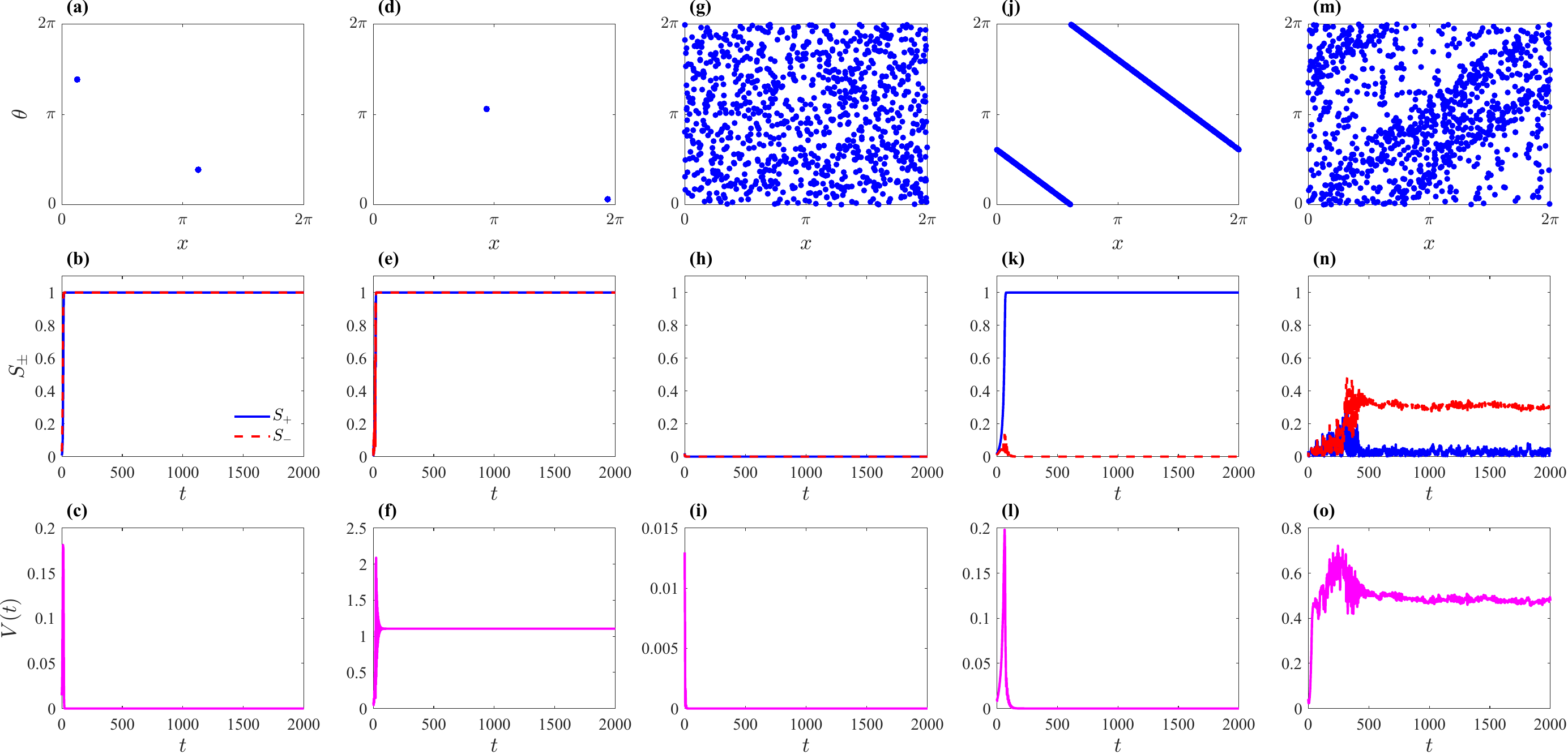}
     \caption{{\bf Representative collective states obtained for $\tau=2.5$.} Each column corresponds to a different collective state: static $\pi$ state, active $\pi$ state, async state, phase wave state, and unsteady state, from left to right. The top row shows the distribution of swarmalators in the $(x,\theta)$ plane, the middle row shows the time evolution of the order parameters $S_+$ and $S_-$, and the bottom row shows the time evolution of the velocity. The parameters are: first column, $J=1$ and $K=1$; second column, $J=-3$ and $K=2.0$; third column, $J=-1.0$ and $K=-1.0$; fourth column, $J=2.0$ and $K=-0.5$; fifth column, $J=-1.3$ and $K=-2.8$. Here, we consider $N=1000$.}
    \label{fig:different_patterns}
\end{figure*}

\subsection{Collective states}
By numerically integrating the governing equations \eqref{eq:x_dot} and \eqref{eq:theta_dot}, we find that the system typically converges to five collective states. Figure~\ref{fig:different_patterns} illustrates these states. The top row shows scatter plots in the $(x,\theta)$ plane where each dot represents one swarmalator, the middle row shows the corresponding time evolution of the order parameters $S_+$ and $S_-$, and the bottom row shows the time evolution of the velocity $V(t)$. We can thus summarize the obtained states as follows.
\begin{itemize}
    \item \textit{Static $\pi$ state}: the swarmalators split into two clusters separated by a phase difference of $\pi$ and a spatial distance of $\pi$. In the $(x,\theta)$ plane, this state appears as two well-defined points. The associated order parameters satisfy $(S_+,S_-)=(1,1)$ and the velocity vanishes (see Fig.~\ref{fig:different_patterns}, first column).

    \item \textit{Active $\pi$ state}: this state has the same two-cluster structure as the static $\pi$ state, with a phase and position separation equal to $\pi$. The main difference is that the velocity is nonzero (see Fig.~\ref{fig:different_patterns}, second column). This state appears because of the presence of the time delay; in the absence of delay, it is not observed in the model studied in Ref.~\cite{o2022collective}

    \item \textit{Async state}: the swarmalators are uniformly distributed in both phase and position. As a result, there is no collective order and the order parameters satisfy $(S_+,S_-)=(0,0)$ (see Fig.~\ref{fig:different_patterns}, third column).

    \item \textit{Phase wave state}: the phases and positions are strongly correlated. Consequently, one of the two order parameters is equal to one, while the other one is equal to zero, namely $(S_+,S_-)=(1,0)$ or $(S_+,S_-)=(0,1)$. In general, which of these two configurations is obtained depends on the initial conditions (Fig.~\ref{fig:different_patterns}, fourth column).

    \item \textit{Unsteady state}: this state is characterized by persistent oscillations of the order parameters, which remain between zero and one. The velocity can also be nonzero. These unsteady states arise due to the presence of the time delay and are shown in Fig.~\ref{fig:different_patterns}, fifth column.
\end{itemize}

Figure~\ref{fig:J_K} shows the bifurcation diagrams in the $(J,K)$ plane for the asymmetric delay model. The numerical results and analytical results, indicate that increasing the delay reduces the stability regions of the phase wave, deforms the stability region of the async state  and favors the emergence of the active $\pi$ state. This behavior differs from the symmetric delay case studied in Ref.~\cite{o2026time}. For comparison, Fig.~\ref{fig:symetric_J_K_plot} shows the corresponding bifurcation diagrams for the symmetric delay model. In that case, the delay also promotes the appearance of active $\pi$ state, but it mainly favors the emergence of unsteady states, in contrast to the asymmetric delay case studied here.

\begin{figure*}[ht!]
    \centering
    \begin{tabular}{ccc}
        \begin{overpic}[width=0.34\textwidth]{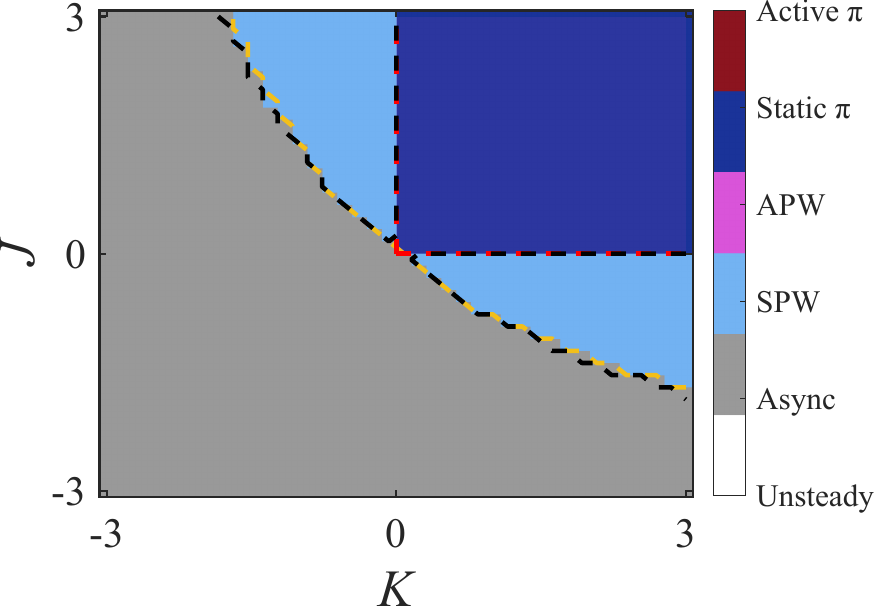}
            \put(4,75){\bfseries (a)}
        \end{overpic}
        &
        \begin{overpic}[width=0.34\textwidth]{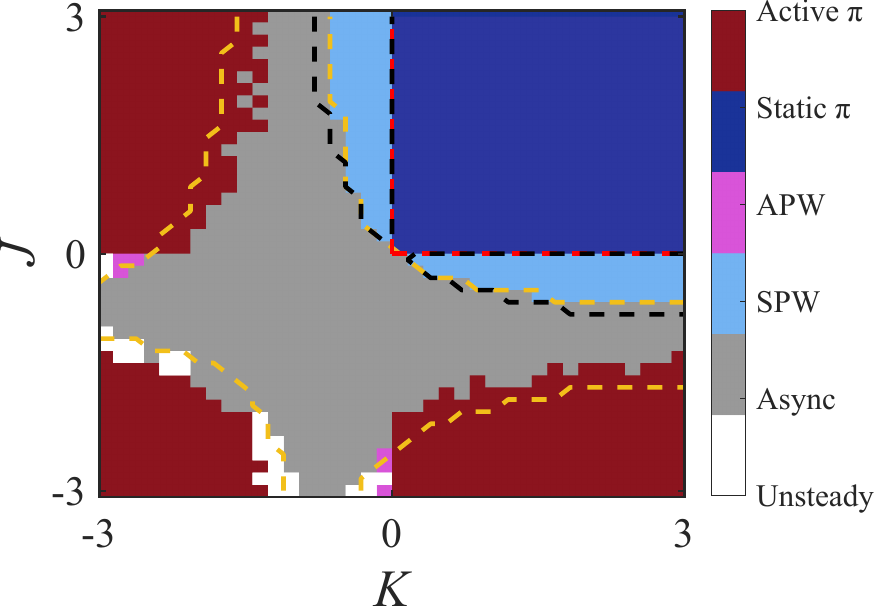}
            \put(4,75){\bfseries (b)}
        \end{overpic}
        &
        \begin{overpic}[width=0.34\textwidth]{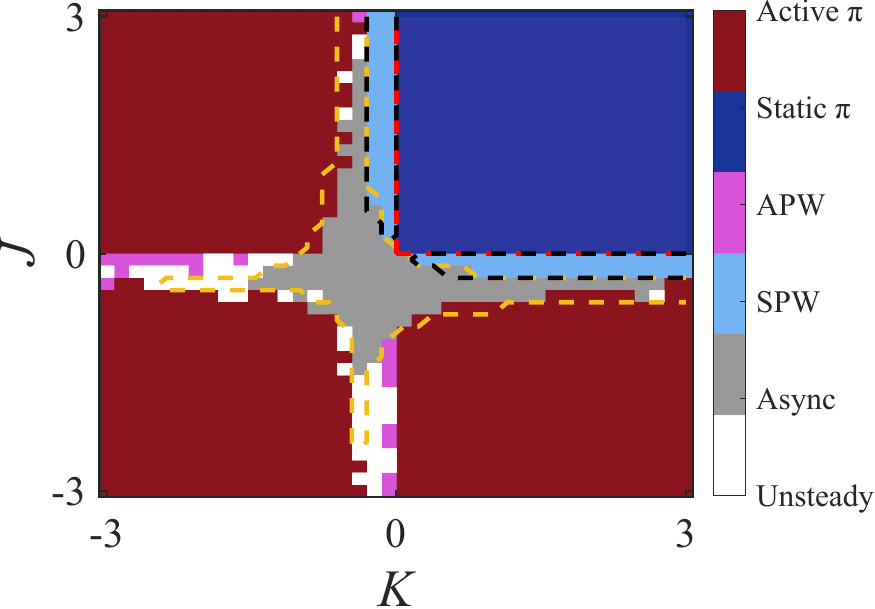}
            \put(4,75){\bfseries (c)}
        \end{overpic}
    \end{tabular}
    \caption{\textbf{Bifurcation diagrams in the $(J,K)$ plane for the asymmetric delay model.}
The delay is fixed to: (a) $\tau=0.5$, (b) $\tau=2.5$, and (c) $\tau=6$.
For the set of initial conditions considered here, with both $x$ and $\theta$
uniformly distributed in $[-\pi,\pi]$, the numerical simulations show that
increasing $\tau$ significantly modifies the stability regions of the collective
states and promotes the emergence of active async dynamics. The threshold
used to distinguish active states from static states is set to $0.1$ for the
mean velocity. The colors indicate the different collective states: dark red corresponds to the
active $\pi$ state, dark blue to the static $\pi$ state, magenta to the active
phase wave (APW), light blue to the static phase wave (SPW), gray to the
async state (Async), and white to unsteady dynamics. The dashed curves
represent the analytically derived stability boundaries: the yellow dashed line
indicates the stability boundary of the async state (Eq.~\eqref{async_1} and Eq.~\eqref{async_2}), the black dashed line
corresponds to the stability boundary of the phase wave state (Eq.~\eqref{phase_wave_boundary}), and the red dashed
line marks the stability boundary of the static $\pi$ state ($J>0,K>0$).}
    \label{fig:J_K}
\end{figure*}

\subsection{Mathematical analysis}
The goal of this section is to analytically determine the conditions on the model parameters ensuring the existence and stability of some of the dynamical states presented in the previous section.

\subsubsection{Stability of the async state}

 We study the linear stability of the spatially and phase-uniform static state shown in panel (g) of Fig.~\ref{fig:different_patterns};  in the continuum limit of infinitely many swarmalator, the latter can be described by the constant density
\begin{equation}
f_0(x,\theta)=\frac{1}{4\pi^2}, \qquad (x,\theta)\in \mathbb{S}^1\times \mathbb{S}^1.
\end{equation}

To parametrize this equilibrium, we introduce Lagrangian labels $(a,b)\in \mathbb{S}^1\times \mathbb{S}^1$ and write
\begin{equation}
x(a,b,t)=a+u(a,b,t), \; \theta(a,b,t)=b+v(a,b,t),
\end{equation}
where $u$ and $v$ are small perturbations. Linearizing the dynamics around the uniform state yields
\begin{widetext}
    \begin{equation}
\partial_t u(a,b,t)=\frac{J}{4\pi^2}\int_0^{2\pi}\int_0^{2\pi}
\Bigl[u(a',b',t-\tau)\cos(a'-a)\cos(b'-b)
-v(a',b',t)\sin(a'-a)\sin(b'-b)\Bigr]\,da'\,db',
\end{equation}
and
\begin{equation}
\partial_t v(a,b,t)=\frac{K}{4\pi^2}\int_0^{2\pi}\int_0^{2\pi}
\Bigl[v(a',b',t-\tau)\cos(a'-a)\cos(b'-b)
-u(a',b',t)\sin(a'-a)\sin(b'-b)\Bigr]\,da'\,db'.
\end{equation}
\end{widetext}

\vspace{0.5cm}

Since the domain is the torus $\mathbb{S}^1\times\mathbb{S}^1$, we expand the perturbations in Fourier series:
\begin{align}
u(a,b,t) &= \sum_{m,n\in\mathbb{Z}} u_{mn}(t)e^{\iota(ma+nb)}, \\
v(a,b,t) &= \sum_{m,n\in\mathbb{Z}} v_{mn}(t)e^{\iota(ma+nb)},
\end{align}

A direct computation shows that the kernels $\cos(a'-a)\cos(b'-b)$ and $\sin(a'-a)\sin(b'-b)$ act nontrivially only on the modes with $|m|=|n|=1$. More precisely we obtain
\begin{equation}
\dot{u}_{mn}(t)=\frac{J}{4}\delta_{\{|m|,1\}}\delta_{\{|n|,1\}}
\bigl(u_{mn}(t-\tau)+mn\,v_{mn}(t)\bigr),
\end{equation}
\begin{equation}
\dot{v}_{mn}(t)=\frac{K}{4}\delta_{\{|m|,1\}}\delta_{\{|n|,1\}}
\bigl(v_{mn}(t-\tau)+mn\,u_{mn}(t)\bigr).
\end{equation}
Hence, all Fourier modes are linearly neutral except the four active modes
\[
(m,n)\in\{(1,1),(1,-1),(-1,1),(-1,-1)\}.
\]
Introducing $s=mn\in\{\pm1\}$, the dynamics of each active mode reduces to
\begin{equation}
\dot{u}(t)=\frac{J}{4}u(t-\tau)+\frac{Js}{4}v(t),
\;
\dot{v}(t)=\frac{K}{4}v(t-\tau)+\frac{Ks}{4}u(t).
\end{equation}

Looking for normal modes of the form
\begin{eqnarray}
u(t)=Ue^{\lambda t}, \qquad v(t)=Ve^{\lambda t},
\end{eqnarray}
with $(U,V)\neq(0,0)$, one obtains the characteristic equation
\begin{equation}
\left(\lambda-\frac{J}{4}e^{-\lambda\tau}\right)
\left(\lambda-\frac{K}{4}e^{-\lambda\tau}\right)
-\frac{JK}{16}=0.
\label{eq:charac_uniform}
\end{equation}
Therefore, the uniform steady state is linearly stable whenever all roots $\lambda$ of \eqref{eq:charac_uniform} satisfy $\Re(\lambda)<0$, up to neutral modes associated with the symmetries of the problem.

In the nondelayed case $\tau=0$, \eqref{eq:charac_uniform} simplifies to
\begin{equation}
\lambda\left(\lambda-\frac{J+K}{4}\right)=0.
\end{equation}
Thus, besides the neutral eigenvalue $\lambda=0$, the nontrivial eigenvalue is
\begin{eqnarray}
\lambda=\frac{J+K}{4},
\end{eqnarray}
and the uniform state is linearly stable if and only if
\begin{equation}
J+K<0.
\end{equation}
For $\tau>0$, the loss of stability is detected by setting $\lambda=\iota\omega$ with $\omega>0$. Separating real and imaginary parts gives
\begin{equation}
\omega^2+\frac{J+K}{4}\omega\sin(\omega\tau)+\frac{JK}{8}\sin^2(\omega\tau)=0,
\end{equation}
\begin{equation}
\cos(\omega\tau)\left(\frac{J+K}{4}\omega+\frac{JK}{8}\sin(\omega\tau)\right)=0.
\end{equation}
The second factor cannot vanish when $\omega\neq0$. Therefore, one must have
\begin{equation}
\cos(\omega\tau)=0,
\end{equation}
which implies
\begin{equation}
\omega\tau=\frac{\pi}{2}+k\pi,
\qquad k\in\mathbb{Z}.
\end{equation}
Since $\omega>0$ and $\tau>0$, the admissible branches are obtained for
$k=0,1,2,\ldots$. The smallest positive value of $\tau$ corresponds to
the first branch $k=0$, namely
\begin{equation}
\omega\tau=\frac{\pi}{2}.
\end{equation}
Hence
\begin{equation}
\omega=\frac{\pi}{2\tau}.
\end{equation}
Substituting into the real part yields
\begin{equation}
\omega^2+\frac{J+K}{4}\omega+\frac{JK}{8}=0.
\end{equation}
Solving for $\omega$ gives
\begin{equation}
\omega_{\pm}
=
\frac{-(J+K)\pm\sqrt{(J+K)^2-8JK}}{8},
\end{equation}
which requires
\begin{equation}
(J+K)^2-8JK\geq 0.
\end{equation}
The corresponding critical delays are
\begin{equation}
\tau_c^{\pm}
=
\frac{\pi}{2\omega_{\pm}}
=
\frac{4\pi}{-(J+K)\pm\sqrt{(J+K)^2-8JK}}.
\end{equation}
Equivalently, using the quadratic relation, we get
\begin{equation}
\tau_c^{\pm}
=
\frac{\pi}{2JK}
\left[-(J+K)\pm\sqrt{(J+K)^2-8JK}\right], JK\neq 0.
\label{async_1}
\end{equation}

\paragraph*{Zero-level curve of the zero-eigenvalue.}

Consider the characteristic equation \ref{eq:charac_uniform} and by substituting $\lambda=0$, we obtain
\begin{eqnarray}
\Delta(0)
=
\left(-\frac{J}{4}\right)\left(-\frac{K}{4}\right)-\frac{JK}{16}
=
\frac{JK}{16}-\frac{JK}{16}
=
0.
\end{eqnarray}
Hence, $\lambda=0$ is always a root of the characteristic equation.

To determine the boundary associated with this root, we require that $\lambda=0$ be a \emph{multiple root}. Therefore, besides $\Delta(0)=0$, we impose
\begin{eqnarray}
\Delta'(0)=0.
\end{eqnarray}
Expanding $\Delta(\lambda)$ gives
\begin{equation}
    \Delta(\lambda)
=
\lambda^2
-\frac{J+K}{4}\lambda e^{-\lambda\tau}
+\frac{JK}{16}e^{-2\lambda\tau}
-\frac{JK}{16}.
\end{equation}
Differentiating with respect to $\lambda$, we obtain
\begin{eqnarray}
\Delta'(\lambda)
=
2\lambda
-\frac{J+K}{4}\frac{d}{d\lambda}\!\left(\lambda e^{-\lambda\tau}\right)
+\frac{JK}{16}\frac{d}{d\lambda}\!\left(e^{-2\lambda\tau}\right),
\end{eqnarray}
which eventually yields
\begin{eqnarray}
\Delta'(\lambda)
=
2\lambda
-\frac{J+K}{4}\left(e^{-\lambda\tau}-\tau\lambda e^{-\lambda\tau}\right)
-\frac{\tau JK}{8}e^{-2\lambda\tau}.
\end{eqnarray}
Evaluating at $\lambda=0$, we get
\begin{eqnarray}
\Delta'(0)
=
-\frac{J+K}{4}-\frac{\tau JK}{8}.
\end{eqnarray}
Therefore, the condition $\Delta'(0)=0$ yields
\begin{eqnarray}
-\frac{J+K}{4}-\frac{\tau JK}{8}=0.
\end{eqnarray}
Thus, the boundary curve in the $(J,K)$-plane associated with the root $\lambda=0$ is
\begin{equation}
    \,2J+2K+\tau JK=0.
    \label{async_2}
\end{equation}

\subsubsection{Stability of the phase wave}

Still considering the model \eqref{eq:x_dot}--\eqref{eq:theta_dot}, we now study the steady phase wave state (see panel (j) of Fig.~\ref{fig:different_patterns}) characterized by
\begin{equation}
x_i^{*} = \frac{2\pi i}{N}, \qquad \theta_i^{*} = \frac{2\pi i}{N}, \qquad i = 1, \dots, N.
\label{eq:etat_base}
\end{equation}

\begin{proposition}
The state \eqref{eq:etat_base} is a steady state of \eqref{eq:x_dot}--\eqref{eq:theta_dot}.
\end{proposition}

\begin{proof}
Let $\Delta_{ji} = \frac{2\pi(j-i)}{N}$. The summand in \eqref{eq:x_dot} equals $\sin\Delta_{ji}\cos\Delta_{ji} = \frac{1}{2}\sin(2\Delta_{ji})$. Setting $m = j-i$, we obtain
\begin{eqnarray}
\sum_{j=1}^N \sin\Delta_{ji}\cos\Delta_{ji} = \frac{1}{2}\operatorname{Im}\!\left(\sum_{m=0}^{N-1} e^{4\pi \iota m/N}\right) = 0,
\end{eqnarray}
since the geometric sum with ratio $e^{4\pi \iota /N}\neq 1$ (for $N\geq 3$) vanishes. The computation is identical for \eqref{eq:theta_dot}, so $\dot{x}_i = \dot{\theta}_i = 0$.
\end{proof}

\paragraph*{Linearization.}
We perturb the steady state by setting
\begin{equation}
x_i(t) = \frac{2\pi i}{N} + \xi_i(t), \qquad \theta_i(t) = \frac{2\pi i}{N} + \phi_i(t),
\label{eq:perturbation}
\end{equation}
with $|\xi_i|, |\phi_i| \ll 1$. Expanding to first order using $\sin(a+\varepsilon)\approx\sin a + \varepsilon\cos a$ and $\cos(a+\varepsilon)\approx\cos a - \varepsilon\sin a$, and noting that the zeroth-order terms cancel (as in the proof above), the linearized equations read
\begin{align}
\dot{\xi}_i(t) &= \frac{J}{N}\sum_{j=1}^N \Bigl[ \cos^2\!\Delta_{ji}\,\bigl(\xi_j(t-\tau)-\xi_i(t)\bigr) \nonumber \\
&\quad - \sin^2\!\Delta_{ji}\,\bigl(\phi_j(t)-\phi_i(t)\bigr) \Bigr], \label{eq:lin_xi}\\
\dot{\phi}_i(t) &= \frac{K}{N}\sum_{j=1}^N \Bigl[ \cos^2\!\Delta_{ji}\,\bigl(\phi_j(t-\tau)-\phi_i(t)\bigr) \nonumber \\
&\quad - \sin^2\!\Delta_{ji}\,\bigl(\xi_j(t)-\xi_i(t)\bigr) \Bigr]. \label{eq:lin_phi}
\end{align}
\paragraph*{Fourier decomposition.}
Since the coefficients in \eqref{eq:lin_xi}--\eqref{eq:lin_phi} depend only on $m = j-i$, we seek solutions of the form
\begin{equation}
\xi_i(t) = \hat{\xi}_k\,e^{\lambda t}\,e^{2\pi \iota k i/N}, \qquad \phi_i(t) = \hat{\phi}_k\,e^{\lambda t}\,e^{2\pi \iota k i/N}.
\label{eq:ansatz}
\end{equation}
The relevant Fourier sums are
\begin{align}
S_k^c &:= \frac{1}{N}\sum_{m=0}^{N-1} \cos^2\!\!\left(\frac{2\pi m}{N}\right) e^{2\pi \iota km/N}, \label{eq:Skc}\\
S_k^s &:= \frac{1}{N}\sum_{m=0}^{N-1} \sin^2\!\!\left(\frac{2\pi m}{N}\right) e^{2\pi \iota km/N}. \label{eq:Sks}
\end{align}
Using the identities $\cos^2\alpha = \frac{1}{2}+\frac{1}{4}e^{2\iota\alpha}+\frac{1}{4}e^{-2\iota\alpha}$ and $\sin^2\alpha = \frac{1}{2}-\frac{1}{4}e^{2\iota\alpha}-\frac{1}{4}e^{-2\iota\alpha}$, together with the geometric sum $\frac{1}{N}\sum_{m=0}^{N-1}e^{2\pi \iota\ell m/N}=\delta_{\ell\equiv 0\pmod{N}}$, one obtains for $N\geq 5$:
\begin{align}
S_k^c &= \frac{1}{2}\delta_{k,0} + \frac{1}{4}\delta_{k,2} + \frac{1}{4}\delta_{k,N-2}, \\ S_k^s &= \frac{1}{2}\delta_{k,0} - \frac{1}{4}\delta_{k,2} - \frac{1}{4}\delta_{k,N-2}.
\label{eq:Skc_Sks}
\end{align}
Substituting \eqref{eq:ansatz} into \eqref{eq:lin_xi}--\eqref{eq:lin_phi} and using $\frac{1}{N}\sum_j\cos^2\!\Delta_{ji} = \frac{1}{2}$, the modal equations are
\begin{align}
\lambda\hat\xi_k &= J\!\left(S_k^c e^{-\lambda\tau} - \tfrac{1}{2}\right)\hat\xi_k - J\!\left(S_k^s - \tfrac{1}{2}\right)\hat\phi_k, \label{eq:modale_xi}\\
\lambda\hat\phi_k &= K\!\left(S_k^c e^{-\lambda\tau} - \tfrac{1}{2}\right)\hat\phi_k - K\!\left(S_k^s - \tfrac{1}{2}\right)\hat\xi_k. \label{eq:modale_phi}
\end{align}

\paragraph*{Modes $k\neq 0,2,N-2$.}
For these modes, $S_k^c = S_k^s = 0$, and \eqref{eq:modale_xi}--\eqref{eq:modale_phi} reduce to
\begin{eqnarray}
\lambda \begin{pmatrix}\hat\xi_k \\ \hat\phi_k\end{pmatrix} = \frac{1}{2}\begin{pmatrix} -J & J \\ K & -K\end{pmatrix} \begin{pmatrix}\hat\xi_k \\ \hat\phi_k\end{pmatrix}.
\end{eqnarray}
The eigenvalues of this matrix are
\begin{equation}
\lambda_1 = 0, \qquad \lambda_2 = -\frac{J+K}{2}.
\end{equation}
Thus, these modes are stable if and only if $J + K > 0$.

\paragraph*{Mode $k=0$.}
For $k=0$, we have $S_0^c = S_0^s = \tfrac{1}{2}$, and therefore both coupling terms vanish. 
The equations decouple and become
\begin{equation}
\lambda\hat{\xi}_0 = \frac{J}{2}\bigl(e^{-\lambda\tau}-1\bigr)\hat{\xi}_0,
\;
\lambda\hat{\phi}_0 = \frac{K}{2}\bigl(e^{-\lambda\tau}-1\bigr)\hat{\phi}_0 .
\end{equation}
Each equation can be written in the form
\begin{eqnarray}
(\lambda + a)e^{\lambda\tau} = a,
\end{eqnarray}
with $a=\tfrac{J}{2}$ for the $\hat{\xi}_0$ mode and $a=\tfrac{K}{2}$ for the $\hat{\phi}_0$ mode. 
Using the Lambert $W$ function, the complete sets of characteristic roots are therefore
\begin{equation}
\lambda_n^{(\xi)}=
\frac{1}{\tau}W_n\!\left(\frac{J\tau}{2}e^{J\tau/2}\right)-\frac{J}{2},
\qquad n\in\mathbb{Z},
\end{equation}
and
\begin{equation}
\lambda_n^{(\phi)}=
\frac{1}{\tau}W_n\!\left(\frac{K\tau}{2}e^{K\tau/2}\right)-\frac{K}{2},
\qquad n\in\mathbb{Z}.
\end{equation}
Thus, the characteristic roots are given by the different branches of the Lambert $W$ function.

\vspace{0.5cm}

\paragraph*{Modes $k=2$ and $k=N-2$.}
For $k=2$ (with $N>4$) we obtain $S_2^c = \frac{1}{4}$ and $S_2^s = -\frac{1}{4}$. Setting
\begin{equation}
\mu(\lambda) := \frac{e^{-\lambda\tau}}{4} - \frac{1}{2},
\label{eq:mu}
\end{equation}
the modal equations become
\begin{eqnarray}
\lambda\hat\xi = J\mu(\lambda)\,\hat\xi + \frac{3J}{4}\,\hat\phi, \;
\lambda\hat\phi = K\mu(\lambda)\,\hat\phi + \frac{3K}{4}\,\hat\xi.
\end{eqnarray}
A nontrivial solution exists if and only if the determinant vanishes, giving the characteristic equation
\begin{equation}
\lambda^2 - (J+K)\mu(\lambda)\,\lambda + JK\!\left[\mu(\lambda)^2 - \frac{9}{16}\right] = 0.
\label{eq:carac_k2}
\end{equation}

\medskip
\noindent\textit{Case $\tau=0$.}
Here $\mu(0) = -\frac{1}{4}$, and \eqref{eq:carac_k2} reduces to the polynomial
\begin{equation}
\lambda^2 + \frac{J+K}{4}\,\lambda - \frac{JK}{2} = 0.
\label{eq:poly_tau0}
\end{equation}
For all roots of \eqref{eq:poly_tau0} to satisfy $\Re(\lambda)<0$, the coefficients must be positive, namely
\begin{equation}
\frac{J+K}{4} > 0
 \Longleftrightarrow 
J+K > 0,
-\frac{JK}{2} > 0
\Longleftrightarrow
JK < 0.
\end{equation}
The stability condition for mode $k=2$ without delay is therefore
\begin{equation}
J + K > 0 \quad \text{and} \quad JK < 0.
\label{eq:stab_k2_tau0}
\end{equation}

\medskip
\noindent\textit{Case $\tau>0$: Hopf bifurcation.}
We seek a purely imaginary eigenvalue $\lambda = \iota \omega$, $\omega > 0$. Writing $\mu(\iota\omega) = \mu_R - \iota \mu_I$ with
\begin{eqnarray}
\mu_R = \frac{\cos\omega\tau - 2}{4}, \qquad \mu_I = \frac{\sin\omega\tau}{4},
\end{eqnarray}
and separating real and imaginary parts of \eqref{eq:carac_k2}, we obtain
\begin{align}
-\omega^2 - (J+K)\omega\,\mu_I + JK\!\left(\mu_R^2 - \mu_I^2 - \tfrac{9}{16}\right) &= 0, \label{eq:R}\\
-(J+K)\omega\,\mu_R - 2JK\,\mu_R\mu_I &= 0. \label{eq:I}
\end{align}
From \eqref{eq:I}, factoring out $\frac{\cos\omega\tau - 2}{4}\neq 0$ yields
\begin{equation}
\sin\omega\tau = -\frac{2(J+K)}{JK}\,\omega.
\label{eq:sin_wt}
\end{equation}
Substituting \eqref{eq:sin_wt} into \eqref{eq:R} and using $\cos^2\omega\tau + \sin^2\omega\tau = 1$ to eliminate $\sin\omega\tau$, one arrives at the quadratic equation in $c = \cos\omega\tau$:
\begin{eqnarray}
c^2 - 2c - 3 + \frac{4(J^2+K^2)}{J^2K^2}\,\omega^2 = 0,
\end{eqnarray}
whose solution is
\begin{equation}
\cos(\omega\tau)= -1-\frac{4\omega^2}{JK}.
\label{eq:cos_wt}
\end{equation}
Enforcing $\cos^2\omega\tau + \sin^2\omega\tau = 1$ with \eqref{eq:sin_wt} and \eqref{eq:cos_wt}, setting $u = \omega_c^2$, one finds
\begin{eqnarray}
u\!\left(4u + (J+K)^2 + 2JK\right) = 0,
\end{eqnarray}
so the non-trivial Hopf frequency satisfies
\begin{equation}
\omega_c^2 = -\frac{(J+K)^2 + 2JK}{4}.
\label{eq:omega_c}
\end{equation}
This is real and positive only if
\begin{equation}
(J+K)^2 + 2JK < 0,
\label{eq:omega_c_positif}
\end{equation}
which requires $J$ and $K$ to have opposite signs with $|JK|$ sufficiently large. The corresponding critical delay is
\begin{equation}
\tau_c^{(n)}=\frac{1}{\omega_c}\left(\mathrm{Arg}(c_c+ \iota s_c)+2\pi n\right),
\qquad n\in\mathbb{Z}.
\end{equation}
where
\begin{equation}
c_c=-1-\frac{4\omega_c^2}{JK}, 
\qquad 
s_c=-\frac{2(J+K)}{JK}\,\omega_c.
\end{equation}
Thus, 
\begin{equation}      
\tau_c(J,K)
=
\frac{4}{\sqrt{\Tilde{A}}}
\arctan\!\left(
\frac{J+K}{\sqrt{\Tilde{A}}}
\right),
\label{phase_wave_boundary}
\end{equation}
where $\Tilde{A} = -(J^2+4JK+K^2)$.

We also consider the case \(\lambda=0\), since stability may change when a root crosses the imaginary axis at the origin. Substituting \(\lambda=0\) into the characteristic equation gives $JK = 0$, so we can have \(J=0\) or \(K=0\). These lines therefore define stationary critical boundaries in the parameter plane.


\vspace{0.5cm}

\subsubsection{Stability of the active $\pi$ state}

We consider a two-cluster solution of the form
\begin{equation}
(x_i(t),\theta_i(t))=
\begin{cases}
(x_s+\omega_1 t,\theta_s+\omega_2 t), & i\in G_1,\\[2pt]
(x_s+\pi+\omega_1 t,\theta_s+\pi+\omega_2 t), & i\in G_2\, ,
\end{cases}
\label{eq:two_cluster_ansatz}
\end{equation}
where $G_1$ and $G_2$ form a partition of $\{1,\dots,N\}$ (see panel (d) of Fig.~\ref{fig:different_patterns}). Substituting the latter into the equations of motion and using that the cluster shift is either $0$ or $\pi$, one finds that all interaction terms are identical. This yields the existence conditions
\begin{equation}
\omega_1=-J\sin(\omega_1\tau),
\qquad
\omega_2=-K\sin(\omega_2\tau).
\label{eq:two_cluster_existence}
\end{equation}

To study linear stability, we perturb the solution as
\begin{eqnarray}
x_i(t)=x_s+\chi_i+\omega_1 t+\delta x_i(t),
\\
\theta_i(t)=\theta_s+\chi_i+\omega_2 t+\delta\theta_i(t),
\end{eqnarray}
with $\chi_i\in\{0,\pi\}$ and $|\delta x_i|,|\delta\theta_i|\ll1$. A first-order expansion of the delayed interaction terms gives
\begin{align}
\dot{\delta x_i}(t)
&=
\frac{J}{N}\sum_{j=1}^N
\cos(\omega_1\tau)\big(\delta x_j(t-\tau)-\delta x_i(t)\big), \\
\dot{\delta\theta_i}(t)
&=
\frac{K}{N}\sum_{j=1}^N
\cos(\omega_2\tau)\big(\delta\theta_j(t-\tau)-\delta\theta_i(t)\big).
\end{align}
Expanding the perturbations in discrete Fourier modes and separating the
uniform and nonuniform components, one obtains for all $m\neq 0$
\begin{equation}
\mu_m=-J\cos(\omega_1\tau),
\qquad
\lambda_m=-K\cos(\omega_2\tau).
\label{eq:nonuniform_modes}
\end{equation}
Hence the nonuniform modes are stable if and only if
\begin{equation}
J\cos(\omega_1\tau)>0,
\qquad
K\cos(\omega_2\tau)>0.
\label{eq:nonuniform_stability}
\end{equation}

For the uniform mode $m=0$, the characteristic roots can be obtained in closed form. One finds
\begin{equation}
\mu_n=
\frac{1}{\tau}
W_n\!\Big(J\tau\cos(\omega_1\tau)\,e^{J\tau\cos(\omega_1\tau)}\Big)
-
J\cos(\omega_1\tau),
\end{equation}
and similarly
\begin{equation}
\lambda_n=
\frac{1}{\tau}
W_n\!\Big(K\tau\cos(\omega_2\tau)\,e^{K\tau\cos(\omega_2\tau)}\Big)
-
K\cos(\omega_2\tau),
\end{equation}
where $W_n$ denotes the $n$th branch of the Lambert $W$ function.

The above analysis provides the existence and stability conditions of the active \(\pi\)-state. We now complement these analytical results by examining how the corresponding frequency branches are organized as the delay \(\tau\) varies. More precisely, the branches are obtained by solving the transcendental equations in \eqref{eq:two_cluster_existence} for each value of \(\tau\), and then reconstructing the corresponding frequencies \(\omega_x\) and \(\omega_\theta\). Introducing
\begin{eqnarray}
y_x=\omega_x\tau,
\qquad
y_\theta=\omega_\theta\tau,
\end{eqnarray}
the existence conditions become
\begin{equation}
    y_x+J\tau\sin(y_x)=0,
\;
y_\theta+K\tau\sin(y_\theta)=0, \label{eqqq}
\end{equation}
so that the frequencies are recovered through
\begin{eqnarray}
\omega_x=\frac{y_x}{\tau},
\qquad
\omega_\theta=\frac{y_\theta}{\tau}.
\end{eqnarray}

For \(J=1\), we consider two values of the phase-coupling parameter, namely \(K=-2\) and \(K=1\). The resulting branch diagrams are shown in Fig.~\ref{fig:branches}. Panels (a,b) correspond to \(K=-2\), whereas panels (c,d) correspond to \(K=1\). The left panels, (a,c), display the branches of \(\omega_x\), while the right panels, (b,d), display those of \(\omega_\theta\). The stability of each branch is determined from condition \eqref{eq:nonuniform_stability}. In the figure, stable branches are represented by solid blue curves, whereas unstable branches are shown by red dashed curves. As \(\tau\) increases, additional branches appear because of the periodicity of the sine function, leading to the coexistence of several active \(\pi\)-states with alternating stability.

\begin{remark}
When several stable branches coexist for the same parameter values, the asymptotic frequency selected by the dynamics may depend on the initial condition. In that case, each branch may possess its own basin of attraction, and the observed long-time behavior is determined by the basin in which the initial condition lies.
\end{remark}

For $J>0$, the trivial solution $y_x=0$ is always stable. For $J<0$, a stable nontrivial branch exists only if $y_x \in (\pi/2,\pi)$ (up to symmetry), which requires
 \begin{eqnarray}
 -J\tau>\frac{\pi}{2}.
 \end{eqnarray}
Thus, stable solutions for the $x$ dynamics exist when
\begin{equation}
 J>0
 \qquad\text{or}\qquad
 J<-\frac{\pi}{2\tau}.
\label{eq:J_stable_region}
\end{equation}
By the same argument, the $\theta$ dynamics admits stable solutions when
\begin{equation}
 K>0
 \qquad\text{or}\qquad
 K<-\frac{\pi}{2\tau}.
\label{eq:K_stable_region}
\end{equation}
Hence the stable region in the $(J,K)$ plane lies inside
\begin{equation}
 \left((-\infty,-\tfrac{\pi}{2\tau})\cup(0,\infty)\right)
 \times
 \left((-\infty,-\tfrac{\pi}{2\tau})\cup(0,\infty)\right).
\label{eq:JK_stable_region}
\end{equation}
The portion of this region that lies in the first quadrant (i.e., $J>0$ and $K>0$) corresponds to the solution $\omega_x = \omega_{\theta} = 0$, which indicates the static $\pi$ state.

\begin{figure}[t!]
    \centering
    \begin{tabular}{ccccccccc}
        \begin{overpic}[width=0.48\linewidth]{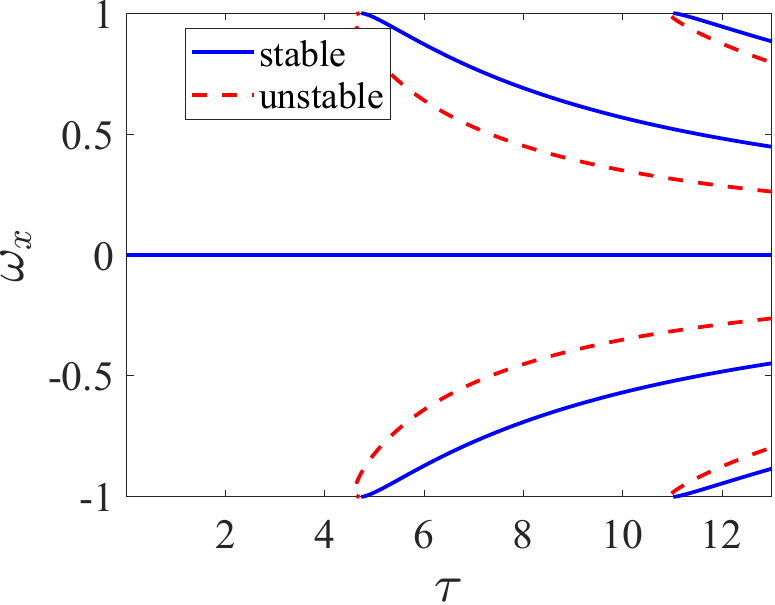}
            \put(4,81){\bfseries (a)}
        \end{overpic}
        &&&&
        \begin{overpic}[width=0.48\linewidth]{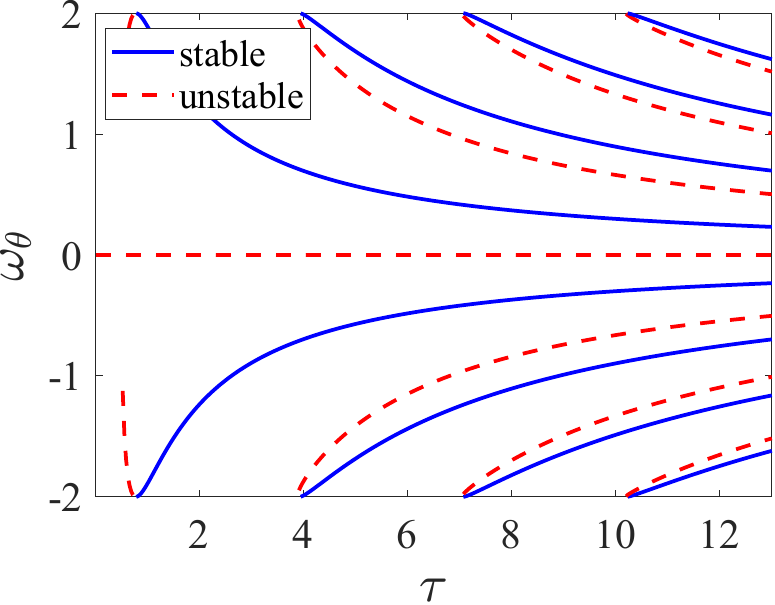}
            \put(4,81){\bfseries (b)}
        \end{overpic}
        \\
        \begin{overpic}[width=0.48\linewidth]{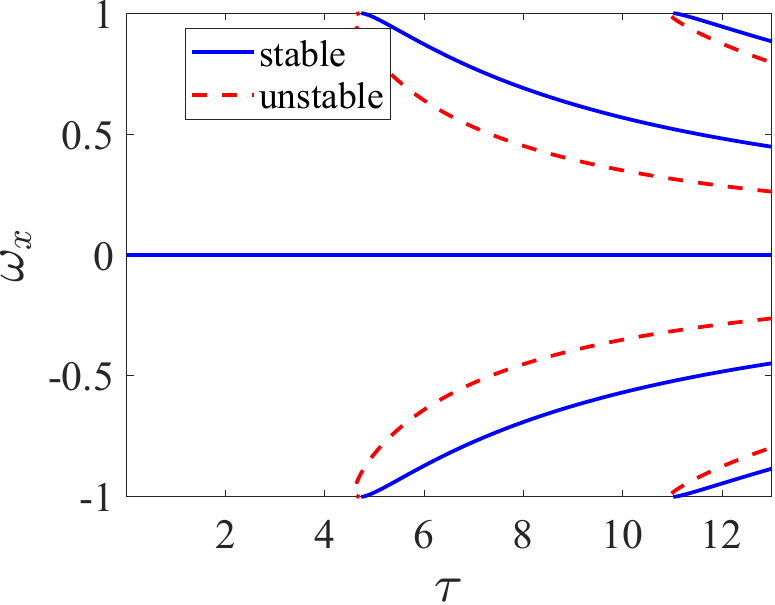}
            \put(4,81){\bfseries (c)}
        \end{overpic}
        &&&&
        \begin{overpic}[width=0.48\linewidth]{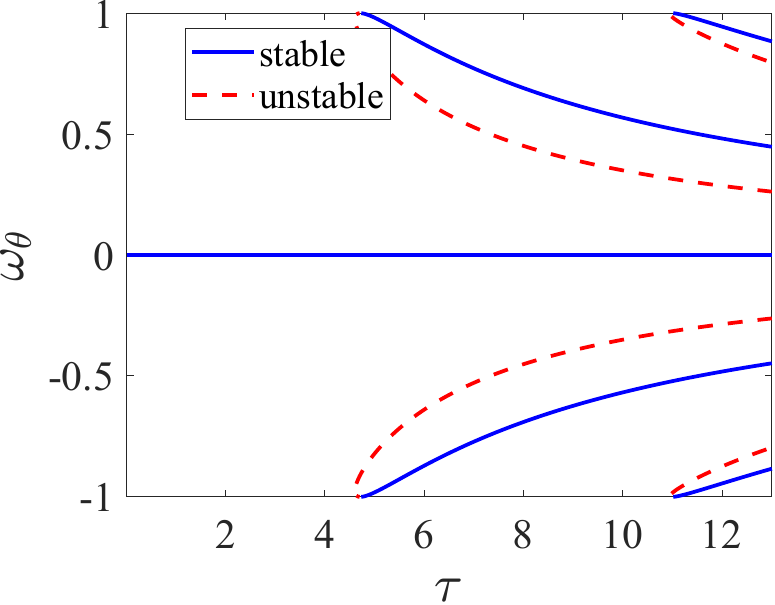}
            \put(4,81){\bfseries (d)}
        \end{overpic}
        \\
    \end{tabular}
\caption{{\bf Numerically computed frequency branches as functions of the delay.} Here we use \(J=1\). Panels (a,b) correspond to \(K=-2\), whereas panels (c,d) correspond to \(K=1\). The left panels, (a,c), display \(\omega_x\), while the right panels, (b,d), display \(\omega_\theta\). Solid blue curves denote stable branches, while red dashed curves denote unstable ones.}
    \label{fig:branches}
\end{figure}

\begin{figure*}[ht!]
    \centering
    \begin{tabular}{ccc}
        \begin{overpic}[width=0.3\linewidth]{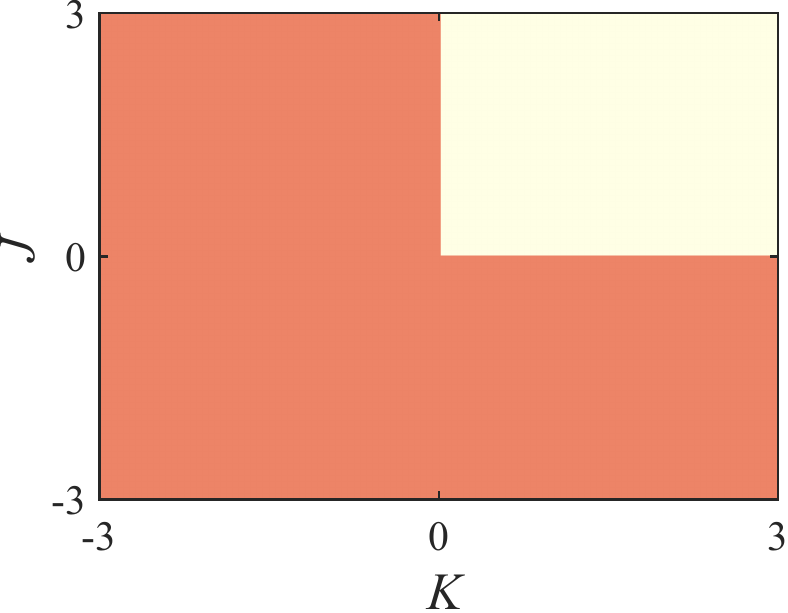}
            \put(4,85){\bfseries (a)}
        \end{overpic}
        &
        \begin{overpic}[width=0.3\linewidth]{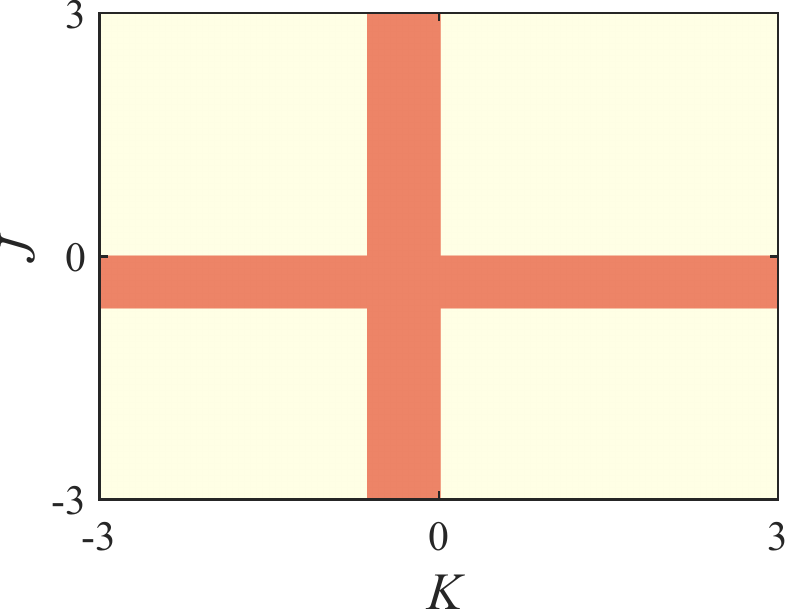}
            \put(4,85){\bfseries (b)}
        \end{overpic}
        &
        \begin{overpic}[width=0.3\linewidth]{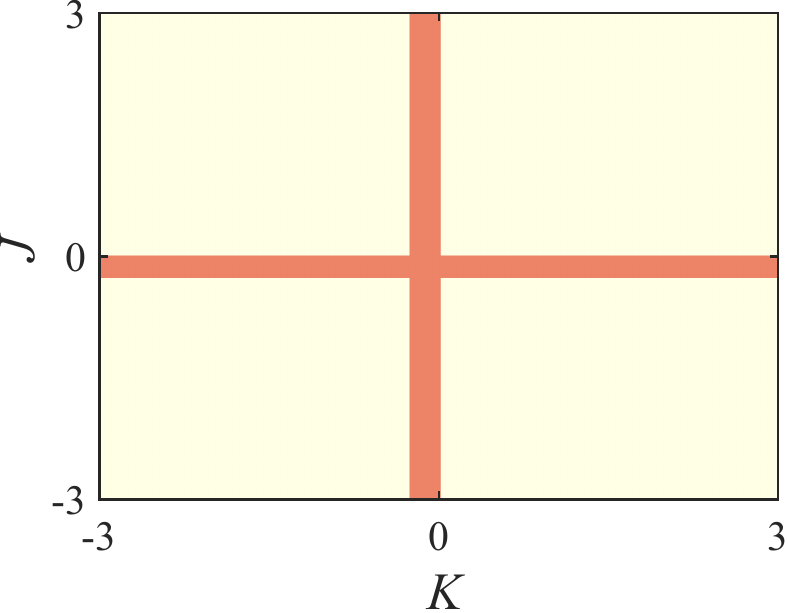}
            \put(4,85){\bfseries (c)}
        \end{overpic}
    \end{tabular}
    \caption{{\bf Stability regions of the $\pi$ states in the $(J,K)$ plane for the asymmetric delay model.}
    The light beige regions correspond to parameter values for which at least one stable pair of frequency branches $(\omega_x,\omega_\theta)$ exists, whereas the orange regions indicate instability. Panels (a), (b), and (c) correspond respectively to $\tau=0.5$, $\tau=2.5$, and $\tau=6$.}
    \label{fig:pi_state_region_stability}
\end{figure*}

Figure~\ref{fig:pi_state_region_stability} shows the stability regions of the $\pi$-state in the $(J,K)$ parameter plane. The light beige regions correspond to parameter values for which the $\pi$-state is stable, whereas the orange regions indicate instability. Panels~\ref{fig:pi_state_region_stability}(a), \ref{fig:pi_state_region_stability}(b), and \ref{fig:pi_state_region_stability}(c) correspond respectively to $\tau=0.5$, $\tau=2.5$, and $\tau=6$. These stability diagrams are constructed by scanning the values of $J$ and $K$ separately. For each value of $J$, we solve the existence equation in~\eqref{eqqq} to determine all admissible frequency branches $\omega_x$. Similarly, for each value of $K$, we compute the corresponding branches $\omega_\theta$. Each branch is then tested using the stability conditions derived above. Therefore, for a given pair $(J,K)$, the  $\pi$ state is classified as stable whenever there exists at least one pair of frequencies $(\omega_x,\omega_\theta)$ satisfying both the existence equations and the stability conditions. Compared with the symmetric model shown in Fig.~\ref{fig:pi_state_region_stability_sym}, we observe that increasing the delay $\tau$ enlarges the stability region of the $\pi$ state. This indicates that the delay can promote the stabilization of active $\pi$ states over a wider range of coupling parameters.

\subsection{Bistability}
The stability analyses in the previous section establishes parameter regions where the async state is linearly stable and the active $\pi$ state exist, but it does not by itself preclude these regions from overlapping. Indeed, direct numerical integration of Eqs.~\eqref{eq:x_dot}-\eqref{eq:theta_dot} reveals that such overlap does occur: for a range of coupling parameters, the two states coexist as stable attractors, and the asymptotic dynamics selected by the system depends on the initial condition.

To quantify this phenomenon, we fix $K=1.35$ and $\tau = 2.5$ and vary $J$ across the region near the boundary of these states.  For each value of $J$, we integrate the system from an ensemble of random initial conditions, with $x_i$ and $\theta_i$ drawn independently and uniformly from $[-\pi, \pi]$, and classify the long-time state reached in each realization using the order parameters and the mean velocity. The fractions of realizations converging to the active $\pi$ and async states are plotted as functions of $J$ in Fig.~\ref{bistability}.
\begin{figure}[htp]
    \centering
    \includegraphics[width=0.9\linewidth]{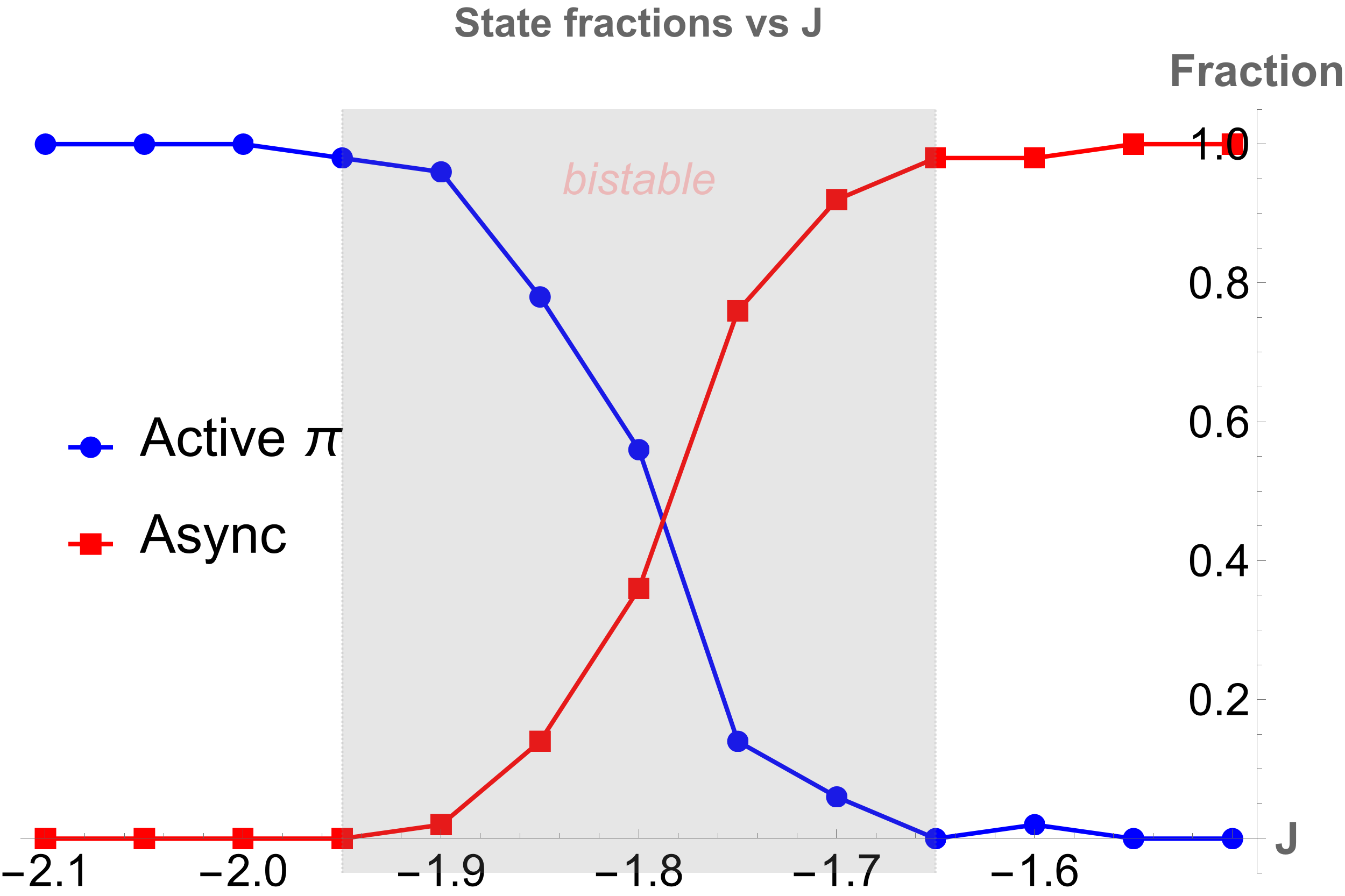}
    \caption{{\bf Bistability between the active $\pi$ and async states.} Fraction of random initial conditions, drawn uniformly from $[-\pi,\pi]^2$, converging to the active $\pi$ state (blue circles) and the async state (red squares) as a function of $J$, at fixed $K = 1.35$ and $\tau = 2.5$. The shaded region $-1.95 \lesssim J \lesssim -1.65$ marks the bistable window in which both states are observed with nonzero probability. Simulation parameters: $(N, t, dt) = (1000, 200, 0.01)$. At each point, $50$ realizations have been taken to calculate the fractions of the two emerging states.}
    \label{bistability}
\end{figure}

Three regimes emerge. For $J \lesssim 1.95$, every initial condition converges to the active $\pi$ state: the async state has a vanishing basin of attraction at these parameters. For $J \gtrsim -1.65$, the situation is reversed and almost every realization converges to the async state. In the intermediate window $-1.95 \lesssim J \lesssim -1.65$, both attractors capture a finite fraction of the trajectories, and the fraction of active $\pi$ realizations decreases monotonically with $J$ while the async fraction grows correspondingly. At $J \approx -1.8$, the two basins are approximately balanced, with roughly half of the random initial conditions selecting each state. This value is close to the theoretical critical $J_c = -1.86$ (derived by solving Eq.~\eqref{async_1} with $K=1.35$, and $\tau = 2.5$) where the static async state is predicted to lose its stability. This intermediate window is therefore a genuine bistable region, in which the asymptotic behavior is not determined by the coupling parameters alone but also by the initial data.

\section{Discussion}

We have presented a theoretical study of delay-coupled swarmalators, with a particular focus on the asymmetric case. We derived the stability boundaries of the three main collective states of the model: the async state, the phase wave state, and the $\pi$ states. This analysis allowed us to construct the full phase diagram in the $(J,K)$ parameter space. We observe that the delay can play the role of a control parameter. Moreover, depending on the way the asymmetric model is defined, the presence of delay can promote ordered states, such as the active $\pi$ state, or generate unsteady states that are absent in the model without delay. For instance, when the delay is introduced only in the ``sine" interaction term, increasing the delay enlarges the stability region of the $\pi$ state. Such dynamics may be relevant to populations of interacting organisms, including predator-prey systems \cite{van2017disentangling}. Similar anti-phase patterns can also be found in biological motor systems, for example in left-right neural activity during swimming or locomotion. Also, as discussed in ~\cite{o2026time}, unsteady states may be observable in populations of vinegar eels, sperm cells, or other real-world swarmalators.

Future work could extend our study to the regime where the natural frequencies and velocities are non-identical. Another possible direction would be to extend the model to two or three spatial dimensions. It would also be interesting to investigate the effects of external forcing, or higher-order interactions with delay.

\section*{Data Availability Statement}
This study did not involve the use of external datasets. The computational codes supporting the findings of this work are available from the authors upon reasonable request.

\section*{Acknowledgments}
G.K.S. acknowledges financial support through Canadian Neuroanalytics scholars (CNS) program.

\section*{Author Contributions}
\noindent

R.T.D.: conceptualization, data curation, formal analysis, investigation, methodology, software, validation, visualization, and writing - original draft. 
G.K.S.: conceptualization, methodology, software, validation, visualization, and writing - original draft. 
T.C.: funding acquisition, project administration, resources, supervision, validation, and writing - review \& editing. 
All authors read and approved the manuscript.

\appendix

\setcounter{equation}{0}
\renewcommand{\theequation}{A\arabic{equation}}
\setcounter{figure}{0}
\renewcommand{\thefigure}{B\arabic{figure}}

\section{Another variant of asymmetric delay}
\label{appA}
There are several ways to define asymmetric delay. For the following asymmetric model, we observe several dynamical regimes, including the async state, the phase wave state, the static async state, and the synchronized state with an unsteady profile. We briefly show below how the linear stability analysis can also be carried out for this alternative asymmetric model:
\begin{align}
\dot{x}_i &= \frac{J}{N}\sum_{j=1}^N \sin \bigl(x_j-x_i\bigr)\cos \bigl(\theta_j(t-\tau)-\theta_i\bigr), 
\label{eq:x_dynamics} \\
\dot{\theta}_i &= \frac{K}{N}\sum_{j=1}^N \sin \bigl(\theta_j-\theta_i\bigr)\cos \bigl(x_j(t-\tau)-x_i\bigr).
\label{eq:theta_dynamics}
\end{align}

\subsection{ Async state}

We consider the async state
\begin{eqnarray}
f_0(x,\theta)=\frac{1}{4\pi^2},
\qquad (x,\theta)\in \mathbb{S}^1\times \mathbb{S}^1,
\end{eqnarray}
and introduce Lagrangian labels $(a,b)\in \mathbb{S}^1\times \mathbb{S}^1$ such that
\begin{eqnarray}
x(a,b,t)=a+u(a,b,t),\; \theta(a,b,t)=b+v(a,b,t),
\end{eqnarray}
with $u,v$ small perturbations. Linearizing around $f_0$ gives
\begin{widetext}
    \begin{align}
\partial_t u(a,b,t)
&=
\frac{J}{4\pi^2}
\int_0^{2\pi}\!\!\int_0^{2\pi}
\Big[
u(a',b',t)\cos(a'-a)\cos(b'-b)-
v(a',b',t-\tau)\sin(a'-a)\sin(b'-b)
\Big]\,
da'\,db',
\\
\partial_t v(a,b,t)
&=
\frac{K}{4\pi^2}
\int_0^{2\pi}\!\!\int_0^{2\pi}
\Big[
v(a',b',t)\cos(a'-a)\cos(b'-b) -
u(a',b',t-\tau)\sin(a'-a)\sin(b'-b)
\Big]\,
da'\,db'.
\end{align}
\end{widetext}

Expanding in Fourier series,
\begin{align}
u(a,b,t)&=\sum_{m,n\in\mathbb Z}u_{mn}(t)e^{\iota (ma+nb)}, \\ 
v(a,b,t)&=\sum_{m,n\in\mathbb Z}v_{mn}(t)e^{\iota (ma+nb)},
\end{align}
one finds that only the modes $|m|=|n|=1$ are active. For these modes,
\begin{align}
    \dot u_{mn}(t)&=\frac{J}{4}\Big(u_{mn}(t)+mn\,v_{mn}(t-\tau)\Big), \\
    \dot v_{mn}(t)&=\frac{K}{4}\Big(v_{mn}(t)+mn\,u_{mn}(t-\tau)\Big).
\end{align}
Setting $s=mn\in\{\pm1\}$ and seeking normal modes
\begin{eqnarray}
u(t)=Ue^{\lambda t},\qquad v(t)=Ve^{\lambda t},
\end{eqnarray}
yields the characteristic equation
\begin{eqnarray}
\left(\lambda-\frac{J}{4}\right)\left(\lambda-\frac{K}{4}\right)
-\frac{JK}{16}e^{-2\lambda\tau}=0.
\end{eqnarray}
For $\tau=0$, this reduces to
\begin{eqnarray}
\lambda\left(\lambda-\frac{J+K}{4}\right)=0,
\end{eqnarray}
so the async state is stable iff
\begin{eqnarray}
J+K<0,
\end{eqnarray}
For $\tau>0$, setting $\lambda=\iota \omega$ with $\omega > 0$ gives
\begin{eqnarray}
\omega^2=\frac{JK}{8}\sin^2(\omega\tau),
\;
\frac{J+K}{4}\,\omega=\frac{JK}{16}\sin(2\omega\tau).
\end{eqnarray}
Combining these two relations with $\sin^2(\omega\tau)+\cos^2(\omega\tau)=1$ leads to
\begin{eqnarray}
16\omega^2+J^2+K^2=0,
\end{eqnarray}
which is impossible for $\omega > 0$. Hence there is no nontrivial Hopf bifurcation for the async state in this model.

Therefore, the stability boundary is determined only by the zero eigenvalue. Imposing that $\lambda=0$ be a double root gives
\begin{eqnarray}
\Delta'(0)=0
\Longrightarrow
-\frac{J+K}{4}+\frac{\tau JK}{8}=0,
\end{eqnarray}
that is,
\begin{eqnarray}
2J+2K-\tau JK=0.
\end{eqnarray}

However, this condition describes the stability boundary only in the parameter region where no other unstable eigenvalue is already present. In particular, it is relevant in the mixed-sign sector $JK<0$, where the loss of stability can occur through a real eigenvalue crossing at $\lambda=0$. In that case, the async state is stable  when
\begin{eqnarray}
2J+2K-\tau JK<0.
\end{eqnarray}
By contrast, when $J>0$ and $K>0$, a positive real eigenvalue already exists, so the condition $\Delta'(0)=0$ is not the global stability boundary in that sector.


\subsection{Phase wave  }
The phase wave state
\begin{eqnarray}
x_i^{(0)}=\theta_i^{(0)}=\frac{2\pi i}{N}, \qquad i=1,\dots,N,
\end{eqnarray}
is again a steady state. We perturb it as
\begin{eqnarray}
x_i(t)=\frac{2\pi i}{N}+\xi_i(t),
\theta_i(t)=\frac{2\pi i}{N}+\phi_i(t), 
\end{eqnarray}
where $|\xi_i|,|\phi_i|\ll1$. At linear order, writing $\Delta_{ji}=2\pi (j-i)/N$, we obtain
\begin{align}
\dot{\xi}_i(t)
&=
\frac{J}{N}\sum_{j=1}^N
\Big[
\cos^2\!\Delta_{ji}\,(\xi_j(t)-\xi_i(t)) \nonumber \\
&\quad -
\sin^2\!\Delta_{ji}\,(\phi_j(t-\tau)-\phi_i(t))
\Big], \\
\dot{\phi}_i(t)
&=
\frac{K}{N}\sum_{j=1}^N
\Big[
\cos^2\!\Delta_{ji}\,(\phi_j(t)-\phi_i(t)) \nonumber \\
&\quad -
\sin^2\!\Delta_{ji}\,(\xi_j(t-\tau)-\xi_i(t))
\Big].
\end{align}
Using the discrete Fourier ansatz
\begin{eqnarray}
\xi_i(t)=\hat{\xi}_k e^{\lambda t} e^{2\pi \iota k i/N},
\;
\phi_i(t)=\hat{\phi}_k e^{\lambda t} e^{2\pi \iota k i/N},
\end{eqnarray}
together with
\begin{align}
S_k^c&=\frac1N\sum_{m=0}^{N-1}\cos^2\!\Big(\frac{2\pi m}{N}\Big)e^{2\pi \iota km/N},
\\
S_k^s&=\frac1N\sum_{m=0}^{N-1}\sin^2\!\Big(\frac{2\pi m}{N}\Big)e^{2\pi \iota km/N},
\end{align}
the modal equations become
\begin{align}
\lambda \hat{\xi}_k
&=
J\Big(S_k^c-\frac12\Big)\hat{\xi}_k
-
J\Big(e^{-\lambda\tau}S_k^s-\frac12\Big)\hat{\phi}_k, \\
\lambda \hat{\phi}_k
&=
K\Big(S_k^c-\frac12\Big)\hat{\phi}_k
-
K\Big(e^{-\lambda\tau}S_k^s-\frac12\Big)\hat{\xi}_k .
\end{align}
As in the main text, only the modes $k=0,2,N-2$ are exceptional.

For the  modes $k\neq 0,2,N-2$, one has $S_k^c=S_k^s=0$, hence
\begin{eqnarray}
\lambda
\begin{pmatrix}
\hat{\xi}_k \\ \hat{\phi}_k
\end{pmatrix}
=
\frac12
\begin{pmatrix}
-J & J \\
K & -K
\end{pmatrix}
\begin{pmatrix}
\hat{\xi}_k \\ \hat{\phi}_k
\end{pmatrix},
\end{eqnarray}
with eigenvalues
\begin{eqnarray}
\lambda_1=0, \qquad \lambda_2=-\frac{J+K}{2}.
\end{eqnarray}
Therefore these modes are stable iff
\begin{eqnarray}
J+K>0.
\end{eqnarray}

For the  mode $k=0$, since $S_0^c=S_0^s=1/2$, the characteristic equation is
\begin{eqnarray}
\lambda^2-\frac{JK}{4}\bigl(1-e^{-\lambda\tau}\bigr)^2=0.
\end{eqnarray}
Thus, if $JK<0$, purely imaginary roots $\lambda= \iota \omega$ satisfy
\begin{eqnarray}
\omega_c=\sqrt{-JK},
\;
\tau_n=\frac{(2n+1)\pi}{\sqrt{-JK}},
\; n\in\mathbb{N},
\end{eqnarray}
whereas for $JK>0$ the zero root becomes multiple when
\begin{eqnarray}
JK\,\tau^2=4.
\end{eqnarray}

For the modes $k=2$ and $k=N-2$, one has
\begin{eqnarray}
S_2^c=\frac14,\qquad S_2^s=-\frac14,
\end{eqnarray}
so that the characteristic equation becomes
\begin{eqnarray}
\Bigl(\lambda+\frac{J}{4}\Bigr)\Bigl(\lambda+\frac{K}{4}\Bigr)
-
JK\Bigl(\frac12+\frac14 e^{-\lambda\tau}\Bigr)^2=0.
\end{eqnarray}
For $\tau=0$, this reduces to
\begin{eqnarray}
\lambda^2+\frac{J+K}{4}\lambda-\frac{JK}{2}=0,
\end{eqnarray}
hence the nondelayed condition
\begin{eqnarray}
J+K>0,\qquad JK<0.
\end{eqnarray}

For $\tau>0$, setting $\lambda=\iota \omega$ shows that a Hopf crossing is possible only if $JK<0$. If $u_c=\omega_c^2>0$ is a positive root of
\begin{multline}
256u^3+32(J^2+8JK+K^2)u^2
+ \\ \bigl(J^4+16J^3K+66J^2K^2+16JK^3+K^4\bigr)u
+2J^3K^3=0,
\end{multline}
then the corresponding critical delays are
\begin{eqnarray}
\tau_{c,n}^{(2)}=
\frac{\Arg(c_c+ \iota  s_c)+2\pi n}{\omega_c},
\qquad n\in\mathbb Z,
\end{eqnarray}
where
\begin{eqnarray}
c_c=-1+\sqrt{\frac{-8u_c}{JK}},
\;
s_c=-\frac{2(J+K)\omega_c}{JK(2+c_c)}.
\end{eqnarray}
The same expression holds for the mode $k=N-2$.

\begin{figure*}[ht!]
    \centering
    \begin{tabular}{ccc}
        \begin{overpic}[width=0.3\textwidth]{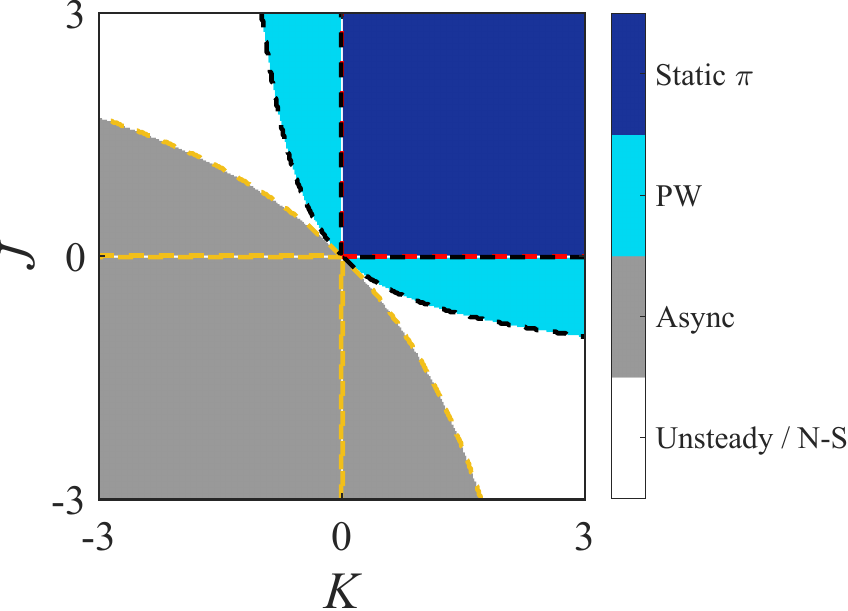}
            \put(4,75){\bfseries (a)}
        \end{overpic}
        &
        \begin{overpic}[width=0.3\textwidth]{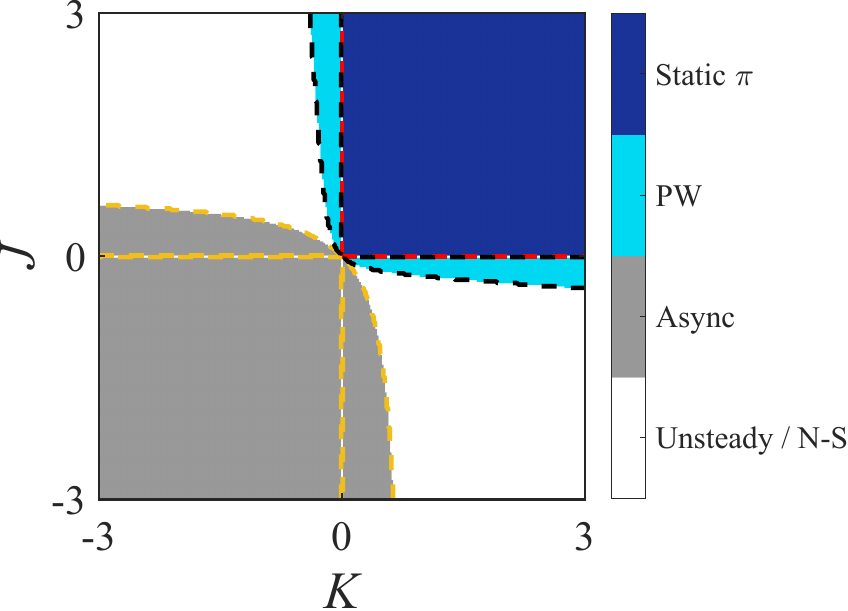}
            \put(4,75){\bfseries (b)}
        \end{overpic}
        &
        \begin{overpic}[width=0.3\textwidth]{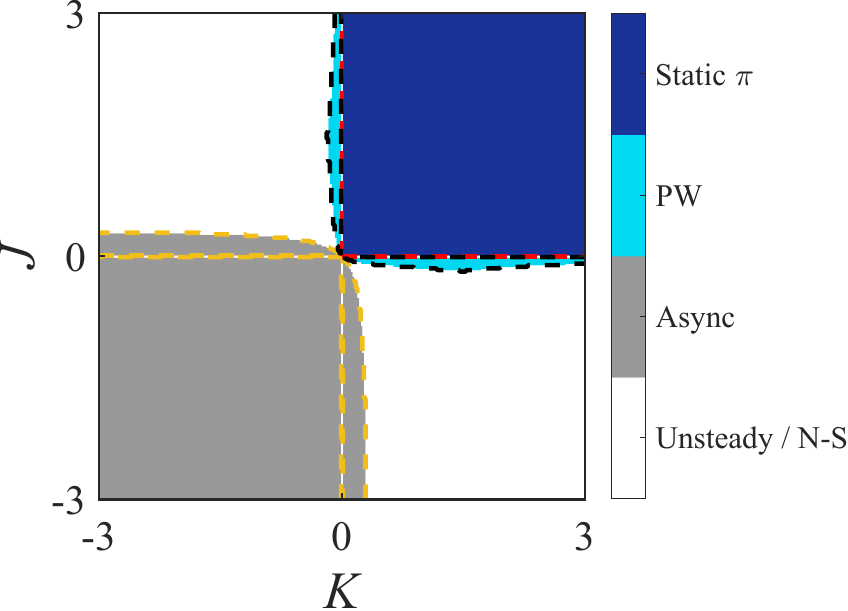}
            \put(4,75){\bfseries (c)}
        \end{overpic} \\
    \end{tabular}
\caption{\textbf{Bifurcation diagrams in the $(J,K)$ plane for the asymmetric delay model with delay in the cosine terms.}
The delay is fixed to: (a) $\tau=0.5$, (b) $\tau=2.5$, and (c) $\tau=6$.
The colored regions represent the stability domains of the main collective states obtained from the analytical stability conditions. 
Cyan corresponds to the phase wave state (PW), dark blue to the static $\pi$ state, gray to the async state (Async), and white to unsteady or non-stationary (N-S) dynamics. 
The dashed curves represent the analytically derived stability boundaries: the yellow dashed line indicates the stability boundary of the async state, the black dashed line corresponds to the stability boundary of the phase-wave state, and the red dashed line marks the stability boundary of the static $\pi$ state.}
    \label{fig:Asymmetric_J_K_cos}
\end{figure*}

\subsection{$\pi$ states}

Following the same approach as for the active $\pi$ state in the main text, we look for a two-cluster solution of the form
\begin{eqnarray}
x_i(t)=x_s+\chi_i+\omega_x t,\\
\theta_i(t)=\theta_s+\chi_i+\omega_\theta t,
\end{eqnarray}
where $\chi_i\in\{0,\pi\}$. Hence, for any pair $(i,j)$,
\begin{eqnarray}
x_j(t)-x_i(t)=\theta_j(t)-\theta_i(t)=\Delta_{ji},
\end{eqnarray}
with
\begin{eqnarray}
 \Delta_{ji}:=\chi_j-\chi_i\in\{0,\pm\pi\}, 
\end{eqnarray}
and
\begin{eqnarray}
\theta_j(t-\tau)-\theta_i(t)=\Delta_{ji}-\omega_\theta\tau, \\
x_j(t-\tau)-x_i(t)=\Delta_{ji}-\omega_x\tau.
\end{eqnarray}
Substituting into the equations of motion gives
\begin{align}
\dot{x}_i&=\frac{J}{N}\sum_{j=1}^N \sin(\Delta_{ji})\cos(\Delta_{ji}-\omega_\theta\tau), \\
\dot{\theta}_i&=\frac{K}{N}\sum_{j=1}^N \sin(\Delta_{ji})\cos(\Delta_{ji}-\omega_x\tau).
\end{align}
Since $\Delta_{ji}\in\{0,\pm\pi\}$, one has $\sin(\Delta_{ji})=0$ for all $(i,j)$, and therefore
\begin{eqnarray}
\dot{x}_i=0,\qquad \dot{\theta}_i=0.
\end{eqnarray}
Then 
\begin{eqnarray}
\omega_x=0,\qquad \omega_\theta=0.
\end{eqnarray}
Thus, the standard active $\pi$ does not appear ; it necessarily reduces to a static $\pi$ state.

We therefore consider the static $\pi$ equilibrium
\begin{eqnarray}
x_i^*=x_s+\chi_i,\; \theta_i^*=\theta_s+\chi_i,
\; \chi_i\in\{0,\pi\},
\end{eqnarray}
and perturb it as
\begin{eqnarray}
x_i(t)=x_i^*+\xi_i(t),
\theta_i(t)=\theta_i^*+\eta_i(t),
\end{eqnarray}
where $ |\xi_i|,|\eta_i|\ll1$. At linear order, using $\sin\Delta_{ji}=0$, $\cos^2\Delta_{ji}=1$, and $\sin^2\Delta_{ji}=0$, one obtains
\begin{eqnarray}
\dot{\xi}_i=\frac{J}{N}\sum_{j=1}^N(\xi_j-\xi_i),
\;
\dot{\eta}_i=\frac{K}{N}\sum_{j=1}^N(\eta_j-\eta_i).
\end{eqnarray}
Writing
\begin{eqnarray}
\bar{\xi}=\frac{1}{N}\sum_{j=1}^N \xi_j,\qquad
\bar{\eta}=\frac{1}{N}\sum_{j=1}^N \eta_j,
\end{eqnarray}
this becomes
\begin{eqnarray}
\dot{\xi}_i=J(\bar{\xi}-\xi_i),\qquad
\dot{\eta}_i=K(\bar{\eta}-\eta_i).
\end{eqnarray}
The uniform modes correspond to the neutral eigenvalue $\lambda=0$, associated with global translations in $x$ and $\theta$. For all nonuniform modes, one finds
\begin{eqnarray}
\lambda_x=-J,\qquad \lambda_\theta=-K.
\end{eqnarray}
Hence the static $\pi$ state is linearly stable iff
\begin{eqnarray}
J>0,\qquad K>0.
\end{eqnarray}

\begin{remark}
We emphasize that the mathematical analysis framework developed in the main body of this work may also be applied to other cases involving asymmetric delays.
\end{remark}

We then visualized the stability regions corresponding to the different observed dynamical regimes, namely the phase wave in cyan, the async state in gray, the static $\pi$ state in blue, and the unsteady, or non-static, state in white, as shown in Fig.~\ref{fig:Asymmetric_J_K_cos}. We can see that increasing the delay affects these stability regions significantly. Figure~\ref{fig:Asymmetric_J_K_cos}(a) corresponds to $\tau = 0.5$, Fig.~\ref{fig:Asymmetric_J_K_cos}(b) corresponds to $\tau = 2.5$, and Fig.~\ref{fig:Asymmetric_J_K_cos}(c) corresponds to $\tau = 6$. As $\tau$ increases, the stability region of the phase wave decreases, as does that of the async state, while the stability region of the unsteady, or non-static, state increases.

\section{Symmetric delay}
\label{appB}
The symmetric delay model has already been studied in a recent work \cite{o2026time}. The purpose of presenting these results (see Fig. \ref{fig:symetric_J_K_plot}) here is therefore to compare them with those obtained for the asymmetric model studied in the main text.

\vspace{3cm}
\begin{figure*}[ht!]
    \centering
    \begin{tabular}{ccc}
        \begin{overpic}[width=0.3\textwidth]{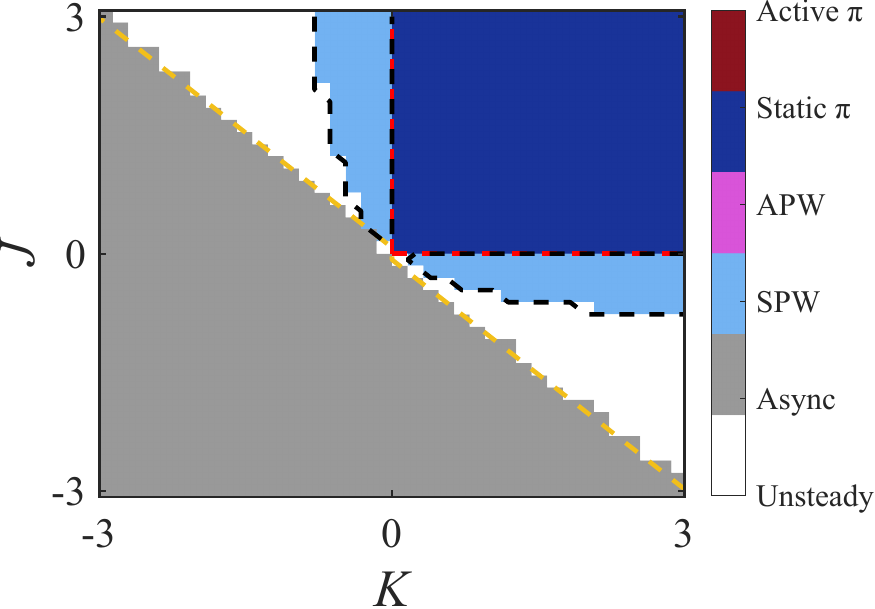}
            \put(4,71){\bfseries (a)}
        \end{overpic}
        &
        \begin{overpic}[width=0.3\textwidth]{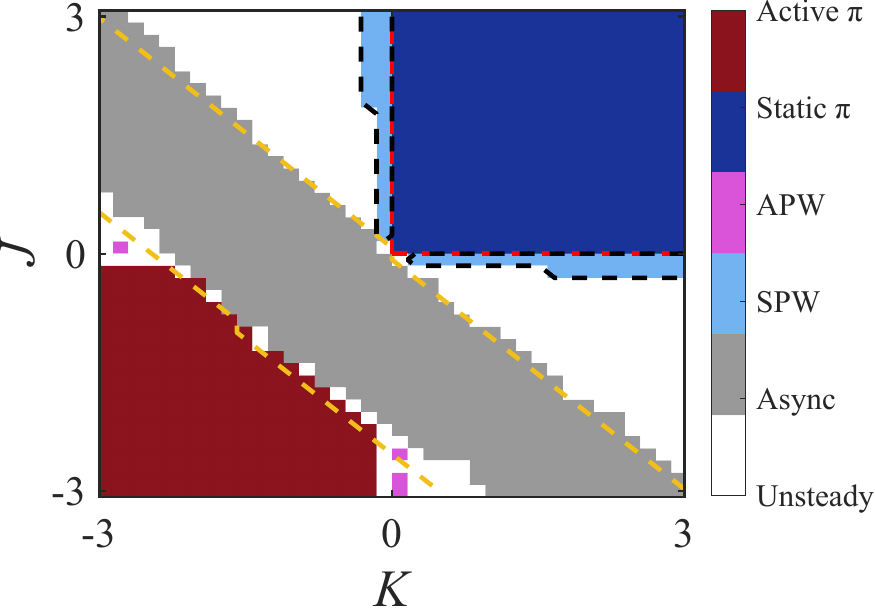}
            \put(4,71){\bfseries (b)}
        \end{overpic}
        &
        \begin{overpic}[width=0.3\textwidth]{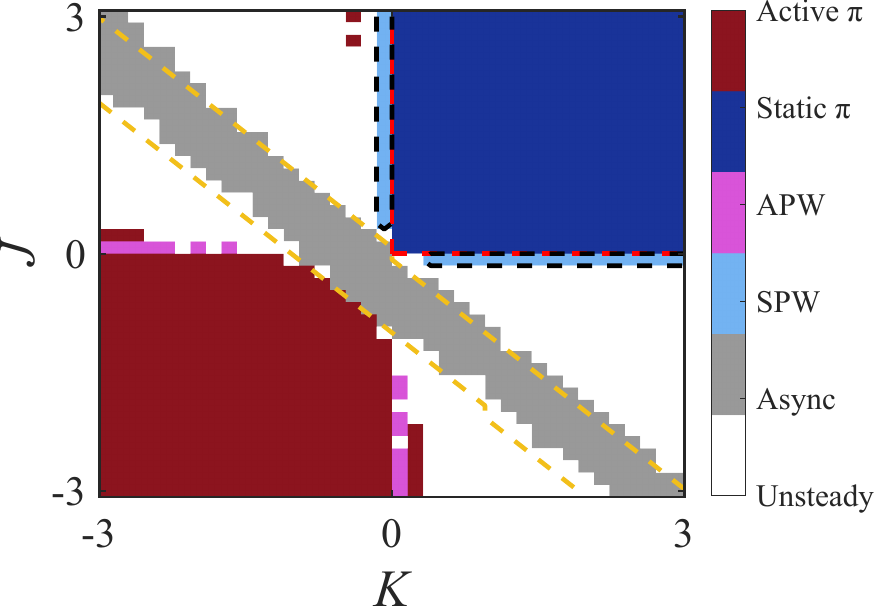}
            \put(4,71){\bfseries (c)}
        \end{overpic}
    \end{tabular}

 \caption{\textbf{Bifurcation diagrams in the $(J,K)$ plane for the symmetric delay model.}
The delay is fixed to: (a) $\tau=0.5$, (b) $\tau=2.5$, and (c) $\tau=6$.
For the set of initial conditions considered here, with both $x$ and $\theta$
uniformly distributed in $[-\pi,\pi]$.  The threshold
used to distinguish active states from static states is set to $0.1$ for the
mean velocity. The colors indicate the different collective states: dark red corresponds to the
active $\pi$ state, dark blue to the static $\pi$ state, magenta to the active
phase wave (APW), light blue to the static phase wave (SPW), gray to the
async state (Async), and white to unsteady dynamics. The dashed curves
represent the analytically derived stability boundaries: the yellow dashed line
indicates the stability boundary of the async state, the black dashed line
corresponds to the stability boundary of the phase-wave state, and the red dashed
line marks the stability boundary of the static $\pi$ state.}
    \label{fig:symetric_J_K_plot}
\end{figure*}
For the analysis of the active $\pi$ state in this model, the existence of a solution requires that
\begin{subequations}
\begin{align}
\omega_1 &= -J\sin(\omega_1\tau)\cos(\omega_2\tau), \label{eq:existence_sym1}\\
\omega_2 &= -K\sin(\omega_2\tau)\cos(\omega_1\tau). \label{eq:existence_sym2}
\end{align}
\end{subequations}

Let
\begin{eqnarray}
C=\cos(\omega_1\tau)\cos(\omega_2\tau),
\;
S=\sin(\omega_1\tau)\sin(\omega_2\tau).
\end{eqnarray}

For the zero mode $m=0$, the characteristic equation is

\begin{equation}
\left[\lambda - JC(x - 1)\right]
\left[\lambda - KC(x - 1)\right]
- JK S^2 (x - 1)^2 = 0.
\end{equation}
\begin{equation}
\text{Where :} \qquad x = e^{-\lambda \tau}.
\end{equation}

For the nonzero modes $m \neq 0$, the characteristic equation is
\begin{eqnarray}
\lambda^2 + (J+K)C\,\lambda + JK(C^2-S^2) = 0.
\end{eqnarray}

Thus, the stability conditions for $m \neq 0$ are
\begin{eqnarray}
(J+K)C > 0,
\qquad
JK(C^2-S^2) > 0.
\end{eqnarray}

The stability of the $\pi$ state in the $(J,K)$ parameter plane is illustrated in
Fig.~\ref{fig:pi_state_region_stability_sym}. The light beige regions correspond to
parameter values for which the $\pi$ state is stable, whereas the orange regions
indicate instability. Panels~\ref{fig:pi_state_region_stability_sym}(a),
\ref{fig:pi_state_region_stability_sym}(b), and
\ref{fig:pi_state_region_stability_sym}(c) correspond respectively to
$\tau=0.5$, $\tau=2.5$, and $\tau=6$.

These figures are obtained by numerically solving the existence equations
\eqref{eq:existence_sym1} and \eqref{eq:existence_sym2} for each pair $(J,K)$,
while exploring the different solution branches. For each solution found,
stability is then determined through the spectral analysis of both the zero mode
$m=0$ and the nonzero modes $m \neq 0$.

From these results, one can observe that the $\pi$ state is predominantly stable
in the quadrants where $J$ and $K$ have the same sign for the symmetric delay.

\begin{figure*}[ht!]
    \centering
    \begin{tabular}{ccccccccc}
        \begin{overpic}[width=0.25\linewidth]{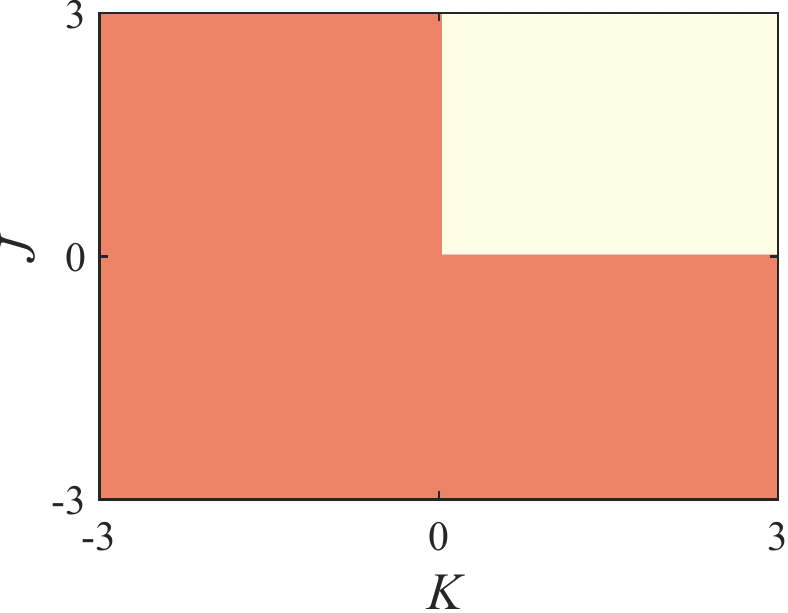}
            \put(4,80){\bfseries (a)}
        \end{overpic}
        &&&&
        \begin{overpic}[width=0.25\linewidth]{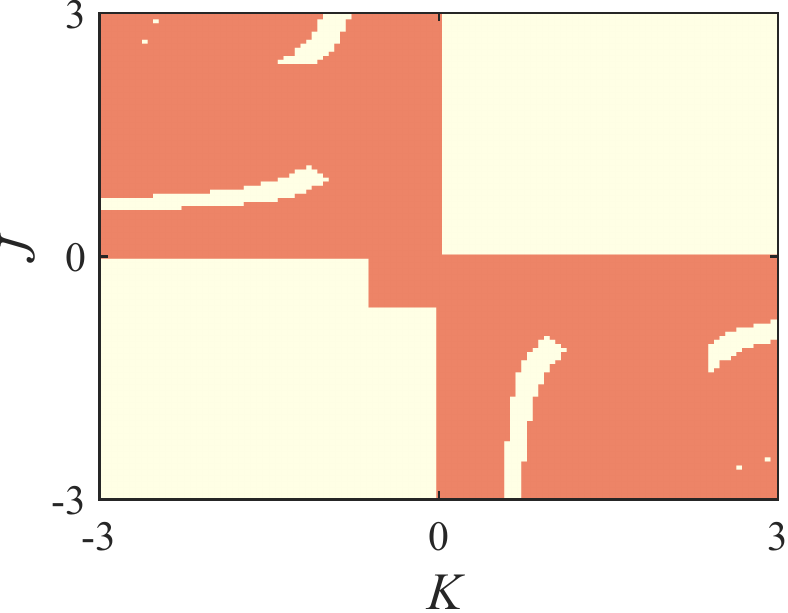}
            \put(4,80){\bfseries (b)}
        \end{overpic}
        &&&&
        \begin{overpic}[width=0.25\linewidth]{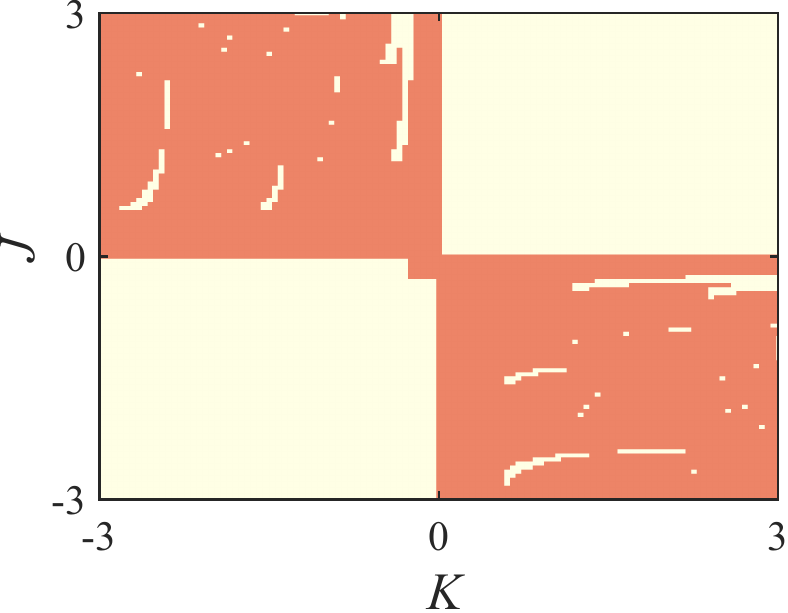}
            \put(4,80){\bfseries (c)}
        \end{overpic}
    \end{tabular}
\caption{{\bf Stability regions of the $\pi$ state in the $(J,K)$ plane for the symmetric delay model.} The light beige regions correspond to parameter values for which at least one stable pair of frequency branches $(\omega_x,\omega_\theta)$ exists, whereas the orange regions indicate instability. Panels (a), (b), and (c) correspond respectively to $\tau=0.5$, $\tau=2.5$, and $\tau=6$.}
    \label{fig:pi_state_region_stability_sym}
\end{figure*}







\bibliographystyle{apsrev}
\bibliography{sample}

\end{document}